\documentclass[final]{siamart171218}

\usepackage[utf8]{inputenc}
\usepackage{amsmath}
\usepackage{amssymb}
\allowdisplaybreaks
\usepackage{tikz}
\usetikzlibrary{arrows,petri,topaths}
\usepackage{tkz-berge}
\usepackage{pgfplots}
\pgfplotsset{compat=1.14}
\usepackage[numbers]{natbib}
\usepackage[caption=false]{subfig}
\hypersetup{colorlinks={true},linkcolor={blue},citecolor=red}
\usepackage[linesnumberedhidden,ruled,algo2e]{algorithm2e}
\usepackage{doi}

\ifpdf
  \DeclareGraphicsExtensions{.eps,.pdf,.png,.jpg}
\else
  \DeclareGraphicsExtensions{.eps}
\fi

\newsiamthm{claim}{Claim}
\def \R{\mathbb R}
\newcommand{\sset}[1]{\left\{ #1\right\}}
\newcommand{\ssets}[1]{\{ #1\}}
\newcommand{\fwh}[1]{\; \left| \; #1 \right.}
\newcommand{\card}[1]{\left| #1 \right|}

\DeclareMathOperator*{\argmin}{argmin}
\newcommand{\algoname}[1]{\ensuremath{\text{\rm\sc #1}}}
\newcommand{\union}{\cup}
\newcommand{\map}{\longrightarrow}
\newcommand{\ifif}{\Longleftrightarrow}
\newcommand{\inters}{\cap}
\newcommand{\then}{\Longrightarrow}
\DeclareMathOperator*{\expect}{\mathbb E}
\newcommand{\vecc}[1]{\ensuremath{\mathbf{#1}}}
\newcommand{\opt}{\ensuremath{\mathrm{OPT}}}
\newcommand{\nashset}[1]{\ensuremath{\mathrm{NE}(#1)}}
\newcommand{\scsum}{C}
\newcommand{\pos}{\mathrm{PoS}}

\headers{The Price of Stability of Weighted Congestion Games}{G. Christodoulou, M. Gairing, Y. Giannakopoulos, and P. G. Spirakis}

\title{The Price of Stability of Weighted Congestion Games\thanks{A preliminary
version of this paper appeared in ICALP'18~\citep{cggs2018}.\funding{Supported
by the Alexander von Humboldt Foundation with funds from the German Federal
Ministry of Education and Research (BMBF), and by EPSRC grants
EP/M008118/1 and EP/L011018/1.}}}

\author{George Christodoulou\thanks{Department of Computer Science, University of Liverpool
  (\email{G.Christodoulou@liverpool.ac.uk}, \email{gairing@liverpool.ac.uk}).}
\and Martin Gairing\footnotemark[2]
\and Yiannis Giannakopoulos\thanks{Chair of Operations Research, TU Munich (\email{yiannis.giannakopoulos@tum.de}).}
\and Paul G. Spirakis\thanks{Department of Computer Science, University of Liverpool; Computer Engineering and Informatics Department, University of Patras  (\email{P.Spirakis@liverpool.ac.uk}).}
}

\ifpdf
\hypersetup{
  pdftitle={The Price of Stability of Weighted Congestion Games},
  pdfauthor={G. Christodoulou, M. Gairing, Y. Giannakopoulos, and P. G. Spirakis}
}
\fi

\begin{document}

\maketitle
\begin{abstract}
We give exponential lower bounds on the Price of Stability (PoS) of weighted
congestion games with polynomial cost functions. In particular, for any positive
integer $d$ we construct rather simple games with cost functions of degree at most
$d$ which have a PoS of at least $\varOmega(\Phi_d)^{d+1}$, where $\Phi_d\sim d/\ln
d$ is the unique positive root of equation $x^{d+1}=(x+1)^d$. This almost closes the
huge gap between $\varTheta(d)$ and $\Phi_d^{d+1}$. Our bound extends also to
network congestion games. We further show that the PoS remains exponential even for
singleton games. More generally, we provide a lower bound of
$\varOmega((1+1/\alpha)^d/d)$ on the PoS of $\alpha$-approximate Nash equilibria for
singleton games. All our lower bounds hold for mixed and correlated equilibria as
well.

On the positive side, we give a general upper bound on the PoS of
$\alpha$-approximate Nash equilibria, which is sensitive to the range $W$ of the
player weights and the approximation parameter $\alpha$. We do this by explicitly
constructing a novel approximate potential function, based on Faulhaber's formula,
that generalizes Rosenthal's potential in a continuous, analytic way. From the
general theorem, we deduce two interesting corollaries. First, we derive the
existence of an approximate pure Nash equilibrium with PoS at most $(d+3)/2$; the
equilibrium's approximation parameter ranges from $\varTheta(1)$ to $d+1$ in a
smooth way with respect to $W$. Secondly, we show that for unweighted congestion
games, the PoS of $\alpha$-approximate Nash equilibria is at most $(d+1)/\alpha$.
\end{abstract}



\section{Introduction}

In the last 20 years, a central strand of research within Algorithmic
Game Theory has focused on understanding and quantifying the
inefficiency of equilibria compared to centralized, optimal
solutions. There are two standard concepts that measure this
inefficiency. The Price of Anarchy (PoA)~\cite{KP99} which takes the
worst-case perspective, compares the worst-case equilibrium with the
system optimum. It is a very robust measure of performance. On the
other hand, the Price of Stability (PoS)~\citep{Schulz2003,ADKTWR04}, which is
also the focus of this work, takes an optimistic perspective, and uses
the best-case equilibrium for this comparison. The PoS is an appropriate
concept to analyse the ideal solution that we would like our protocols
to produce.

The initial set of problems that arose from the Price of Anarchy
theory have now been resolved. The most rich and well-studied among
these models are, arguably, the atomic and non-atomic variants of
congestion games~(see \citep[Ch.~18]{2007a} for a
detailed discussion). This class of games is very descriptive and
captures a large variety of scenarios where users compete for
resources, most prominently routing games. The seminal work of
Roughgarden and Tardos \citep{RT02,RT04} gave the answer for the
non-atomic variant, where each player controls a negligible amount of
traffic.  \citet{AAE05,CK05a} resolved the Price of Anarchy for atomic
congestion games with affine latencies, generalized by
\citet{Aland2011} to polynomials; this led to the development of
Roughgarden's Smoothness Framework~\cite{Rou09} which extended the
bounds to general cost functions, but also distilled and formulated
previous ideas to bound the Price of Anarchy in an elegant, unified
framework. At the computational complexity front, we know that even for simple congestion
games, finding a (pure) Nash equilibrium is a PLS-complete problem~\citep{Fabrikant2004a,Ackermann2008}.

Allowing the players to have different loads, gives rise to the class
of \emph{weighted} congestion games~\citep{Rosenthal1973b}; this is a natural and very
important generalization of congestion games, with numerous
applications in routing and scheduling. Unfortunately though, an
immediate dichotomy between weighted and unweighted congestion games
occurs: the former may \emph{not} even have pure Nash equilibria~\citep{Libman2001,Fotakis2005a,Goemans2005,Harks2012a}; as a matter of fact, it is a strongly NP-hard problem to even determine if that is the case~\citep{Dunkel2008}.
Moreover, in such games
there does not, in general, exist a potential function~\cite{Monderer:1996sp,Harks2011},
which is the main tool for proving equilibrium existence in the
unweighted case.

As a result, a sharp contrast with respect to our understanding of the
two aforementioned inefficiency notions arises.  The Price of Anarchy
has been studied in depth and general techniques for providing tight
bounds are known.  Moreover, the asymptotic behaviour of weighted and
unweighted congestion games with respect to the Price of Anarchy is
identical; it is $\Theta(d/\log d)^{d}$ for both classes when
latencies are polynomials of degree at most $d$~\cite{Aland2011}.

The situation for the Price of Stability though, is completely different. For
unweighted games we have a good understanding\footnote{Much work has been also done
on the PoS for \emph{network design games}, which is though not so closely related
to ours; in such games the cost of using an edge is split equally among players, and
thus cost functions are decreasing, as opposed to our model of congestion games with
nondecreasing latencies. This problem was first studied by \citet{ADKTWR04} who
showed a tight bound of $H_n$, the harmonic number of the number of players $n$, for
directed networks. Finding tight bounds on undirected networks is still a
long-standing open problem (see, e.g., \citep{Fiat2006,Bilo2013,Lee2013}). Recently,
\citet{Bilo2014} (asymptotically) resolved the question for broadcast networks. For
the weighted variant of this problem, \citet{Albers2009} showed a lower bound of
$\varOmega(\log W/\log \log W)$, where $W$ is the sum of the players' weights, while
\citet{Chen2008} an upper bound of $O(\log W/\alpha)$ for $\alpha$-approximate
equilibria (the latter is similar in spirit to our results in
\cref{sec:upper_bounds}). See \cite{Bilo2014,Albers2009} and references therein for
a thorough discussion of those results.} and the values are much lower than the
Price of Anarchy values, and also {\em tight}; approximately 1.577 for affine
functions~\cite{CK05b,Caragiannis2010a}, and $\Theta(d)$~\citep{Christodoulou2015}
for polynomials. For weighted games though there is a huge gap; the current state of
the art lower bound is $\varTheta(d)$ and the upper bound is $\varTheta(d/\ln d)^d$.
These previous results are summarized at the left of \cref{table:results_previous}.

The main focus of this work is precisely to deal with this lack of
understanding, and to determine the Price of Stability of weighted
congestion games.
What makes this problem challenging is that the only general known technique for
showing upper bounds for the Price of Stability is the potential method, which is
applicable only to potential games. In a nutshell, the idea of this method is to use
the global minimizer of Rosenthal's potential \cite{Rosenthal1973a} as an
equilibrium refinement. This equilibrium is also a pure Nash equilibrium and can
serve as an upper bound of the Price of Stability.  Interestingly, it turns out
that, for several classes of potential games, this technique actually provides the
tight answer (see for example~\citep{ADKTWR04, CK05b, Caragiannis2010a,
Christodoulou2015}). However, as already mentioned above, unlike their unweighted
counterparts, weighted congestion games are not potential games;\footnote{For the
special case of weighted congestion games with linear latency functions, a potential
does exist~\citep{Fotakis2005a} and this was used by~\citep{Bilo2017} to provide a
PoS upper bound of $2$.} so, a completely fresh approach is required.
One way to override the aforementioned limitations of non-existence of pure Nash
equilibria, but also their computational hardness, is to consider \emph{approximate}
equilibria. In this direction, \citet{Hansknecht2014} have shown that
$(d+1)$-approximate pure Nash equilibria always exist in weighted congestion games
with polynomial latencies of maximum degree $d$, while, in the negative side, there
exist games that do not have $1.153$-approximate pure Nash equilibria.  Notice here,
that these results do not take into account computational complexity considerations;
if we insist in polynomial-time algorithms for actually finding those equilibria,
then the currently best approximation parameter becomes $d^{O(d)}$
\citep{Caragiannis2011,Caragiannis2015a,Feldotto2017}.

\subsection{Our Results}

We provide lower and upper bounds on the Price of Stability for the
class of weighted congestion games with polynomial latencies
with nonnegative coefficients. We consider both exact and approximate
equilibria.
\renewcommand{\arraystretch}{1.6}
\begin{table}[t]
{\footnotesize
\captionsetup{position=top} 
\caption{The Price of Anarchy and
Stability for unweighted and weighted congestion games, with polynomial latency
functions of maximum degree $d$. $\Phi_d$ is the unique positive solution of
$(x+1)^d=x^{d+1}$ and $\Phi_d=\varTheta(d/\log d)$. Tight answers were known for all
settings, except for the Price of Stability of the weighted case were only trivial
bounds existed. In this paper we almost close this gap by
showing a lower bound of $\varOmega(\Phi_d)^{d+1}$ (\cref{th:PoSLower_weighted}), which remains exponential even for singleton games (\cref{th:PoS_lower_approx}).}
\label{table:results}
\begin{center}
\subfloat[Previous results]{\label{table:results_previous}
\begin{tabular}[b]{r | c | c|}
 \multicolumn{1}{r}{}
 &  \multicolumn{1}{c}{PoA}
 & \multicolumn{1}{c}{PoS} \\
\cline{2-3}
unweighted & $\lfloor \Phi_d\rfloor^{d+1}$ \citep{Aland2011} & $\varTheta(d)$ \citep{Christodoulou2015}\\
\cline{2-3}
weighted & $\Phi_d^{d+1}$ \citep{Aland2011} & $[\varTheta(d),\Phi_d^{d+1}]$\\
\cline{2-3}
\end{tabular}
}
\subfloat[This paper]{\label{table:results_ours}
\begin{tabular}[b]{r| c|}
 \multicolumn{1}{r}{}
 &  \multicolumn{1}{c}{PoS lower bound}\\
 \cline{2-2}
 general & $\varOmega(\Phi_d)^{d+1}$\\
 \cline{2-2}
 singleton & $\varOmega(2^d/d)$\\
  \cline{2-2}
  \parbox[c]{2.5cm}{\raggedleft$\alpha$-approximate equilibria} & $\varOmega((1+1/\alpha)^d/d)$\\
  \cline{2-2}
\end{tabular}
} \end{center}
}
\end{table}
Our lower bounds are summarized at \cref{table:results_ours}.

\paragraph{Lower Bound for Weighted Congestion Games}
In our main result in \cref{th:PoSLower_weighted}, we resolve a
long-standing open problem by providing almost tight bounds
for the Price of Stability of weighted congestion games with
polynomial latency functions. We construct an instance having a Price
of Stability of $\varOmega(\Phi_d)^{d+1}$, where $d$ is the maximum
degree of the latencies and $\Phi_d\sim\frac{d}{\ln d}$ is the unique
positive solution of equation $(x+1)^d=x^{d+1}$.

This bound almost closes the previously huge gap between $\varTheta(d)$ and
$\Phi_d^{d+1}$ for the PoS of weighted congestion games. The previously best lower
and upper bounds were rather trivial: the lower bound corresponds to the PoS results
of \citet{Christodoulou2015} for the unweighted case (and thus, it is also a valid
lower bound for the general weighted case as well) and the upper bound comes from
the Price of Anarchy results of \citet{Aland2011} (PoA, by definition, upper-bounds
PoS). It is important to make clear here that our lower bound still leaves an open
gap for future work: the constant within the base of $\varOmega(\Phi_d)$ in
\cref{th:PoSLower_weighted} produces a lower bound of $(\frac{1}{2} \Phi_d)^{d+1}$,
which is formally a factor of $2^{d+1}$ away from the $\Phi_d^{d+1}$ PoA upper
bound.

Although as mentioned before, weighted congestion
games do not always possess pure equilibria, our lower bound
construction involves a \emph{unique} equilibrium occurring by
iteratively eliminating strongly dominated strategies. As a result,
this lower bound holds not only for pure, but mixed and correlated
equilibria as well.

\paragraph{Singleton Games} Next we switch to the class of singleton
congestion games, where a pure strategy for each player is a {\em
  single} resource. This class is very well-studied as, on one hand,
it abstracts scheduling environments, and on the other, it has very
attractive equilibrium properties; unlike general weighted congestion
games, there exists an (ordinal) lexicographic potential~\citep{Fotakis2009a,Harks2012}, thus
implying the existence of {\em pure} Nash equilibria. It is important
to note that the tight lower bounds for the Price of Anarchy of
general weighted congestion games hold also for the class of
singleton games~\citep{Caragiannis2010a,Bhawalkar2014a,Bilo2017a}.

Even for this special class, we show in
\cref{th:PoS_lower_approx} an exponential lower bound of
$\varOmega(2^d/d)$. The previous best lower and upper bounds were the
same as those of the general case, namely $\varTheta(d)$ and $\Phi_d^{d+1}$,
respectively.
As a matter of fact, this new lower bound comes as a corollary of a
more general result that we show in \cref{th:PoS_lower_approx},
that extends to approximate equilibria and gives a lower bound of
$\varOmega((1+1/\alpha)^d/d)$ on the PoS of $\alpha$-approximate
equilibria, for any (multiplicative) approximation parameter
$\alpha\in[1,d)$.
Setting $\alpha=1$ we recover the special case of exact
equilibria and the aforementioned exponential lower bound on the
standard, exact notion of the PoS.
Notice here that, as we show in \cref{th:PoS_OPT}, the optimal solution (which, in general, is not an equilibrium) itself constitutes a $(d+1)$-approximate equilibrium with a (trivially) optimal PoS of $1$.

\paragraph{Positive Results for Approximate Equilibria}
In light of the above results, in \cref{sec:upper_bounds}, we
turn our attention to identifying environments with more structure or
flexibility with respect to the underlying solution concept, for which
we can hope for improved quality of equilibria. Both our lower bound
constructions discussed above
use players' weights that form a
geometric sequence. In particular the ratio $W$ of the largest over the
smallest weight is equal to $w^n$ (for some $w>1$),
which grows very large as the number of players $n\to\infty$. On the other hand, for games
where the players have equal weights, i.e. $W=1$, we know that the PoS is at most
$d+1$. It is therefore natural to ask how the performance of the good
equilibria captured by the notion of PoS varies with respect to $W$. In \cref{th:PoS_upper_general}, we are able to give a general upper bound for $\alpha$-approximate
equilibria which is sensitive to this parameter $W$ and to $\alpha$.
This general theorem has two immediate, interesting corollaries.

Firstly (\cref{th:PoS_upper_general_exact_only}), by allowing the ratio $W$ to range in $[1,\infty)$, we derive the existence of an
$\alpha$-approximate pure Nash equilibrium with PoS at most $(d+3)/2$; the equilibrium's approximation parameter $\alpha$ ranges from $\varTheta(1)$ to $d+1$ in a smooth way with respect to $W$.
This is
of particular importance in settings where player weights are not very
far away from each other (that is, $W$ is small).
Secondly (\cref{th:PoS_upper_general_unweighted}), by setting $W=1$ and allowing
$\alpha$ to range up to $d+1$, we get an upper bound of $\frac{d+1}{\alpha}$ for the
$\alpha$-approximate PoS of {\em unweighted} congestion games which, to the best of
our knowledge, was not known before, degrading gracefully from $d+1$ (which is the
actual PoS of exact equilibria in the unweighted case~\citep{Christodoulou2015})
down to the optimal value of $1$ if we allow $(d+1)$-approximate equilibria (which
in fact can be achieved by the optimum solution itself; see \cref{th:PoS_OPT}).

\paragraph{Our Techniques}
An advantage of our main lower bound (\cref{th:PoSLower_weighted}) is the simplicity of the underlying
construction, as well as its straightforward adaptation to network games (see
\cref{sec:network})). However, fine-tuning the parameters of the game
(player weights and latency functions), to ensure uniqueness of the
equilibrium at the ``bad'' instance, was a technically involved task.
This was in part due to
the fact that, in order to guarantee uniqueness (via iteratively
dominant strategies), each player interacts with a window of $\mu$
other players. This $\mu$ depends on $d$ in a delicate way (see \cref{fig:plot_b,lemma:def_c}); it has
to be an integer but, at the same time, needs also to balance nicely with the algebraic properties of
$\Phi_d$.

Moreover we needed to provide deeper insights on the asymptotic, analytic behaviour
of $\Phi_d$, and to explore some new algebraic characteristics of $\Phi_d$ (see,
e.g., \cref{lemma:lower_bound_phi_asymptotic}). It is important to keep in mind that
asymptotically $\Phi_d\sim \frac{d}{\ln d}$ (see~\eqref{eq:phiasy}). The fact that
$\Phi_d=\varTheta\left(\frac{d}{\ln d}\right)$ was already known by the work
of~\citet{Aland2011}; here we provide a more refined characterization of $\Phi_d$'s
growth, using the Lambert-W function (see~\eqref{eq:asymptotic_phi_2}).

In order to derive our upper bounds, we need to define a novel {\em
  approximate potential
  function}~\citep{Chen2008,Christodoulou2011a,Hansknecht2014}.
First, in \cref{lemma:approximate_PoS_ratios}, we identify clear algebraic
sufficient conditions for the existence of approximate equilibria with good
social-cost guarantees, and then explicitly define (see
\eqref{eq:potential_mono_def} and \eqref{eq:pot_def_poly} in the proof of
\cref{th:PoS_upper_general}) a function that satisfies them. This continuous
function, which is defined in the entire space of positive reals, essentially
generalizes that of Rosenthal's in a smooth way: by setting $W=\alpha=1$, we recover
exactly the first significant terms of the well known Rosenthal
potential~\citep{Rosenthal1973a} polynomial, with which one can demonstrate the
usual PoS results for the unweighted case (see, e.g.~\citep{CK05b}). The simple,
analytic way in which this function is defined, is the very reason why we can handle
both the approximation parameter $\alpha$ of the equilibrium and the ratio $W$ of
the weights in a smooth manner while at the same time providing good PoS guarantees.

It is important to stress that, by
the purely analytical way in which our approximate potential function
is defined, in principle it can also incorporate more general cost functions than
polynomials; so, we believe that this technique may be of independent
interest. We point towards that direction in \cref{sec:EulerMaclaurin}.

\section{Model and Notation}
Let $\R$ denote the set of real numbers, $\R_{\geq 0}=[0,\infty)$ and $\R_{>0}=(0,\infty)$.

\paragraph{Weighted Congestion Games}
A \emph{weighted congestion game}
consists of a finite, nonempty set of players $N$ and resources (or facilities) $E$.
Each player $i\in N$ has a \emph{weight} $w_i\in\R_{>0}$ and a \emph{strategy set}
$S_i\subseteq 2^E$. Associated with each resource $e\in E$ is a \emph{cost} (or
\emph{latency}) \emph{function} $c_e: \R_{>0}\map \R_{\geq 0}$. In this paper we
mainly focus on polynomial cost functions with maximum degree $d\geq 0$ and
nonnegative coefficients; that is, every cost function is of the form
$c_e(x)=\sum_{j=0}^d a_{e,j} \cdot x^j$, with $a_{e,j}\ge 0$ for all $j$. In the
following, whenever we refer to polynomial cost functions we mean cost functions of
this particular form.

A \emph{pure strategy profile} is a choice of strategies $\vecc{s}= (s_1, s_2,...,
s_n)\in {S}={S}_1\times \cdots \times {S}_n$ by the players. We use the standard
game-theoretic notation $\vecc{s}_{-i}=(s_1,\ldots, s_{i-1},\allowbreak
s_{i+1},\allowbreak \ldots s_n)$, ${S}_{-i}={S}_1\times \cdots \times S_{i-1} \times
S_{i+1} \times \cdots \times S_n$, such that $\vecc{s}=(s_i,\vecc{s}_{-i})$. Given a
pure strategy profile $\vecc{s}$, we define the \emph{load} $x_e(\vecc{s})$ of
resource $e\in E$ as the total weight of players that use resource $e$ on $\vecc s$,
i.e., $x_e(\vecc{s})=\sum_{i\in N: e\in s_i} w_i$. The \emph{cost} player $i$ is
defined by $C_i(\vecc{s})=\sum_{e\in s_i} c_e(x_e(\vecc{s}))$.

A \emph{singleton} weighted congestion game is a special form of congestion games
where the strategies of all players consist only of single resources; that is, for
all players $i\in N$, $\card{s_i}=1$ for all $s_i\in S_i$. In a weighted
\emph{network} congestion games the resources $E$ are given as the edge set of some
directed graph $G=(V,E)$, and each player $i\in N$ has a source $o_i\in V$ and
destination $t_i\in V$ node; then, the strategy set $S_i$ of each player is
implicitly given as the edge sets of all directed $o_i\to t_i$ paths in $G$.

\paragraph{Nash Equilibria} A pure strategy profile $\vecc{s}$ is a pure \emph{Nash
equilibrium} if and only if for every player $i\in N$ and for all $s'_i\in {S}_i$,
we have $C_i(\vecc{s})\leq C_i(s'_i,\vecc{s}_{-i})$. Similarly a strategy profile is
an \emph{$\alpha$-approximate pure Nash equilibrium}, for $\alpha\geq 1$, if
$C_i(\vecc{s})\leq \alpha \cdot C_i(s'_i,\vecc{s}_{-i})$ for all players $i\in N$
and $s'_i\in S_i$. As discussed in the introduction, weighted congestion games do
not always admit pure Nash equilibria. However, by Nash's theorem they have mixed
Nash equilibria. A tuple ${\vecc \sigma}=( \sigma_1, \cdots, \sigma_N)$ of
independent probability distributions over players' strategy sets is a \emph{mixed
Nash equilibrium} if
\[ \expect_{\vecc{s}\sim \sigma}[C_i(\vecc{s})] \leq \expect_{\vecc{s}_{-i} \sim{
  \sigma}_{-i}}[C_i(s_{i}',\vecc{s}_{-i} )] \] holds for every $i\in N$ and $s_i'
\in S_i$. Here $\sigma_{-i}$ is a product distribution of all $\sigma_j$'s with
$j\neq i$, and $\vecc{s}_{-i}$ denotes a strategy profile drawn from this
distribution. We use $\nashset{G}$ to denote the set of all mixed Nash equilibria of
a game $G$.

\paragraph{Social Cost and Price of Stability}
Fix a weighted congestion game $G$. The \emph{social cost} of a pure strategy
profile $\vecc{s}$ is the weighted sum of the players' costs
$$\scsum(\vecc{s})=\sum_{i\in N}w_i \cdot C_i(\vecc{s}) = \sum_{e\in E}
x_e(\vecc{s}) \cdot c_e(x_e(\vecc{s})).$$
Denote by $\opt(G)=\min_{\vecc{s}\in S} \scsum(\vecc{s})$ the \emph{optimum social
cost} over all strategy profiles $\vecc{s}\in {S}$.
Then, the \emph{Price of Stability (PoS)} of $G$ is the social cost of the best-case
Nash equilibrium over the optimum social cost: $$ \pos (G) =
\min_{\sigma\in\nashset{G}}\frac{\expect_{\vecc{s}\sim
\sigma}[\scsum(\vecc{s})]}{\opt(G)}.$$ The Price of Stability of
$\alpha$-approximate Nash equilibria is defined accordingly. The PoS for a class
$\mathcal G$ of games is the worst (i.e., largest) PoS among all games in the class,
that is, $\pos(\mathcal G)=\sup_{G\in \mathcal G} \pos(G)$. For example, our focus
in this paper is determining the Price of Stability for the class $\mathcal G$ of
weighted congestion games with polynomial cost functions.

For brevity, we will sometimes abuse our formal terminology and refer to the ``PoS
of $\vecc s $'' for a specific (approximate) equilibrium $\vecc s $ of a game $G$
(see, e.g., \cref{th:PoS_upper_general}); by that we will mean the approximation
ratio of the social cost of $\vecc s$ to the optimum, i.e., $\frac{C(\vecc
s)}{\opt(G)}$. Clearly, the PoS of \emph{any} such equilibrium $\vecc s$ is a valid
upper bound to the PoS of the entire game $G$.

Finally, notice that, by using a straightforward scaling argument, it is without
loss with respect to the PoS metric to analyse games with player weights in
$[1,\infty)$; if not, divide all $w_i$'s with $\min_i w_i$ and scale cost functions
accordingly.

\section{Lower Bounds}
\label{sec:lower_bounds}
In this section, we present our lower bound constructions. In
\cref{sec:lower_general} we present the general lower bound and then
in \cref{sec:lower_singleton} the lower bound for singleton games.

\subsection{General Congestion Games}
\label{sec:lower_general}
The next theorem presents our main negative result on the Price of
Stability of weighted congestion games with polynomial latencies of
degree $d$, that almost matches the Price of Anarchy upper
bound of $\Phi_{d}^{d+1}$ from \citet{Aland2011}. Our result shows a
strong separation for the Price of Stability between weighted and unweighted
congestion games; the Price of Stability of the latter is at most
$d+1$~\cite{Christodoulou2015}. This is in sharp contrast to the Price of Anarchy
of these two classes, where the respective bounds are essentially the same.

To state our result, we first introduce some notation. Let $\Phi_d\sim \frac{d}{\ln
d}$ be the unique positive root of equation $(x+1)^d=x^{d+1}$ and let $\beta_d$ be a
parameter with $\beta_d\geq 0.38$ for any $d$,
$\lim_{d\to\infty}\beta_d=\frac{1}{2}$ (formally, $\beta_d$ is defined in
\eqref{eq:b_def} below, and a plot of its values can be seen in \cref{fig:plot_b}).

\begin{theorem}
\label{th:PoSLower_weighted} The Price of Stability of weighted congestion games
with polynomial latency functions of degree at most $d\geq 9$ is at least
$(\beta_d\Phi_d)^{d+1}$.
\end{theorem}
We must mention here that the restriction of $d\geq 9$ is without loss: for
polynomial latencies of smaller degrees $d\leq 8$ we can instead apply the simpler
lower-bound instance for singleton games given in \cref{sec:lower_singleton}. To
prove the theorem, we will need the following technical lemma. Its proof can be
found in \cref{app:def_c_proof}.
\begin{lemma}
\label{lemma:def_c} For any positive integer $d$ define
\begin{equation}
\label{eq:c_def} c_d=\frac{1}{d}\left\lfloor
d\frac{\ln(2\cdot \Phi_d+1)-\ln(\Phi_d+1)}{\ln \Phi_d}
\right\rfloor
\end{equation} and
\begin{equation}
\label{eq:b_def}
\beta_d = 1-\Phi_d^{-c_d}.
\end{equation}
Then
\begin{equation}
\label{eq:c_def_lemma_1}
\Phi_d^{d+2}\leq \left(\Phi_d+\frac{1}{\beta_d} \right)^d,
\end{equation} and for all $d\geq 9$,
\begin{equation}
\label{eq:c_def_lemma_2} d\cdot c_d \geq 3,
\qquad 0.38\leq \beta_d \leq \frac{1}{2}
\qquad \text{and}
\qquad
\lim_{d\to\infty} \beta_d=\frac{1}{2}.
\end{equation}
Plots of parameters $c_d$ and $\beta_d$ can be found in \cref{fig:plot_b}.
\begin{figure}
\scalebox{0.84}{
\begin{tikzpicture}
\begin{axis}[width=9cm, title={Plot of $\beta_b$},
    xmin=9, xmax=100, ymin=0.38, ymax=0.5, xtick={10,20,30,40,50,60,70,80,90,100},
    ytick={0.38,0.4,0.45,0.5}, ymajorgrids=true, xmajorgrids=true, grid
    style=dashed, ]
\addplot[ color=blue, mark=*, ] coordinates {
    (4.,0.43016)(5.,0.394577)(6.,0.365517)(7.,0.341221)(8.,0.439911)(9.,0.417652)(10.,0.397938)(11.,0.38033)(12.,0.453611)(13.,0.437096)(14.,0.421947)(15.,0.40799)(16.,0.466516)(17.,0.453274)(18.,0.440886)(19.,0.429266)(20.,0.418342)(21.,0.466979)(22.,0.45646)(23.,0.446478)(24.,0.43699)(25.,0.427959)(26.,0.469646)(27.,0.460881)(28.,0.452485)(29.,0.444436)(30.,0.43671)(31.,0.473231)(32.,0.465695)(33.,0.45843)(34.,0.451421)(35.,0.444653)(36.,0.477177)(37.,0.470555)(38.,0.464139)(39.,0.45792)(40.,0.451888)(41.,0.481223)(42.,0.475307)(43.,0.469554)(44.,0.463958)(45.,0.458511)(46.,0.453207)(47.,0.479888)(48.,0.474668)(49.,0.469575)(50.,0.464604)(51.,0.459752)(52.,0.484271)(53.,0.479489)(54.,0.474813)(55.,0.470239)(56.,0.465764)(57.,0.461383)(58.,0.484038)(59.,0.479713)(60.,0.475475)(61.,0.471319)(62.,0.467245)(63.,0.463249)(64.,0.48431)(65.,0.480359)(66.,0.476479)(67.,0.47267)(68.,0.468928)(69.,0.488635)(70.,0.484933)(71.,0.481294)(72.,0.477716)(73.,0.474196)(74.,0.470734)(75.,0.489233)(76.,0.485805)(77.,0.48243)(78.,0.479107)(79.,0.475835)(80.,0.472612)(81.,0.490045)(82.,0.486851)(83.,0.483703)(84.,0.480601)(85.,0.477542)(86.,0.474527)(87.,0.471553)(88.,0.488022)(89.,0.485071)(90.,0.48216)(91.,0.479288)(92.,0.476454)(93.,0.473656)(94.,0.48928)(95.,0.486502)(96.,0.48376)(97.,0.481051)(98.,0.478376)(99.,0.475734)(100.,0.490597)
    };
\addplot[ color=red ] coordinates {
    (4.,0.245122)(5.,0.277202)(6.,0.301201)(7.,0.319928)(8.,0.335006)(9.,0.347444)(10.,0.357906)(11.,0.366846)(12.,0.374585)(13.,0.38136)(14.,0.387348)(15.,0.392684)(16.,0.397473)(17.,0.401798)(18.,0.405727)(19.,0.409313)(20.,0.412602)(21.,0.415631)(22.,0.41843)(23.,0.421026)(24.,0.423441)(25.,0.425694)(26.,0.427802)(27.,0.429779)(28.,0.431636)(29.,0.433386)(30.,0.435037)(31.,0.436599)(32.,0.438077)(33.,0.43948)(34.,0.440813)(35.,0.442081)(36.,0.44329)(37.,0.444443)(38.,0.445544)(39.,0.446597)(40.,0.447605)(41.,0.448572)(42.,0.449499)(43.,0.450389)(44.,0.451244)(45.,0.452067)(46.,0.45286)(47.,0.453623)(48.,0.454359)(49.,0.45507)(50.,0.455756)(51.,0.456419)(52.,0.45706)(53.,0.45768)(54.,0.458281)(55.,0.458862)(56.,0.459426)(57.,0.459973)(58.,0.460503)(59.,0.461018)(60.,0.461518)(61.,0.462004)(62.,0.462477)(63.,0.462936)(64.,0.463383)(65.,0.463819)(66.,0.464243)(67.,0.464655)(68.,0.465058)(69.,0.46545)(70.,0.465832)(71.,0.466206)(72.,0.46657)(73.,0.466925)(74.,0.467273)(75.,0.467612)(76.,0.467943)(77.,0.468267)(78.,0.468584)(79.,0.468894)(80.,0.469197)(81.,0.469493)(82.,0.469784)(83.,0.470068)(84.,0.470346)(85.,0.470619)(86.,0.470886)(87.,0.471147)(88.,0.471404)(89.,0.471655)(90.,0.471902)(91.,0.472144)(92.,0.472381)(93.,0.472614)(94.,0.472843)(95.,0.473067)(96.,0.473287)(97.,0.473504)(98.,0.473716)(99.,0.473925)(100.,0.47413)
    };
\addplot[color=red] coordinates
{(4.,0.43016)(5.,0.437599)(6.,0.443376)(7.,0.448018)(8.,0.451845)(9.,0.455064)(10.,0.457818)(11.,0.460204)(12.,0.462295)(13.,0.464146)(14.,0.465798)(15.,0.467282)(16.,0.468625)(17.,0.469846)(18.,0.470963)(19.,0.471988)(20.,0.472933)(21.,0.473808)(22.,0.47462)(23.,0.475376)(24.,0.476083)(25.,0.476744)(26.,0.477365)(27.,0.477949)(28.,0.4785)(29.,0.479021)(30.,0.479513)(31.,0.47998)(32.,0.480423)(33.,0.480844)(34.,0.481246)(35.,0.481628)(36.,0.481994)(37.,0.482343)(38.,0.482677)(39.,0.482997)(40.,0.483304)(41.,0.483599)(42.,0.483882)(43.,0.484154)(44.,0.484416)(45.,0.484669)(46.,0.484912)(47.,0.485147)(48.,0.485374)(49.,0.485593)(50.,0.485805)(51.,0.48601)(52.,0.486208)(53.,0.4864)(54.,0.486587)(55.,0.486767)(56.,0.486942)(57.,0.487112)(58.,0.487277)(59.,0.487438)(60.,0.487594)(61.,0.487745)(62.,0.487893)(63.,0.488036)(64.,0.488176)(65.,0.488313)(66.,0.488445)(67.,0.488575)(68.,0.488701)(69.,0.488824)(70.,0.488944)(71.,0.489061)(72.,0.489176)(73.,0.489288)(74.,0.489397)(75.,0.489504)(76.,0.489609)(77.,0.489711)(78.,0.489811)(79.,0.489908)(80.,0.490004)(81.,0.490098)(82.,0.49019)(83.,0.49028)(84.,0.490368)(85.,0.490454)(86.,0.490539)(87.,0.490621)(88.,0.490703)(89.,0.490783)(90.,0.490861)(91.,0.490938)(92.,0.491013)(93.,0.491087)(94.,0.49116)(95.,0.491231)(96.,0.491301)(97.,0.49137)(98.,0.491438)(99.,0.491504)(100.,0.491569)};
\end{axis}
\end{tikzpicture}
\begin{tikzpicture}
\begin{axis}[width=7.2cm, title={Plot of $\mu=d\cdot c_d$},
    xmin=9, xmax=100, ymin=0, ymax=20, xtick={10,20,30,40,50,60,70,80,90,100},
    ytick={3,5,10,15,20}, ymajorgrids=true, xmajorgrids=true, grid style=dashed, ]
\addplot[color=blue] coordinates {
    (4.,2)(5.,2.)(6.,2.)(7.,2.)(8.,3.)(9.,3.)(10.,3.)(11.,3.)(12.,4.)(13.,4.)(14.,4.)(15.,4.)(16.,5.)(17.,5.)(18.,5.)(19.,5.)(20.,5.)(21.,6.)(22.,6.)(23.,6.)(24.,6.)(25.,6.)(26.,7.)(27.,7.)(28.,7.)(29.,7.)(30.,7.)(31.,8.)(32.,8.)(33.,8.)(34.,8.)(35.,8.)(36.,9.)(37.,9.)(38.,9.)(39.,9.)(40.,9.)(41.,10.)(42.,10.)(43.,10.)(44.,10.)(45.,10.)(46.,10.)(47.,11.)(48.,11.)(49.,11.)(50.,11.)(51.,11.)(52.,12.)(53.,12.)(54.,12.)(55.,12.)(56.,12.)(57.,12.)(58.,13.)(59.,13.)(60.,13.)(61.,13.)(62.,13.)(63.,13.)(64.,14.)(65.,14.)(66.,14.)(67.,14.)(68.,14.)(69.,15.)(70.,15.)(71.,15.)(72.,15.)(73.,15.)(74.,15.)(75.,16.)(76.,16.)(77.,16.)(78.,16.)(79.,16.)(80.,16.)(81.,17.)(82.,17.)(83.,17.)(84.,17.)(85.,17.)(86.,17.)(87.,17.)(88.,18.)(89.,18.)(90.,18.)(91.,18.)(92.,18.)(93.,18.)(94.,19.)(95.,19.)(96.,19.)(97.,19.)(98.,19.)(99.,19.)(100.,20.)
    };
\end{axis}
\end{tikzpicture}}
\caption{The values of parameters $\beta_b$ and $c_d$ in \cref{lemma:def_c} and
\cref{th:PoSLower_weighted}, for $d=9,10,\dots,100$.}
\label{fig:plot_b}
\end{figure}
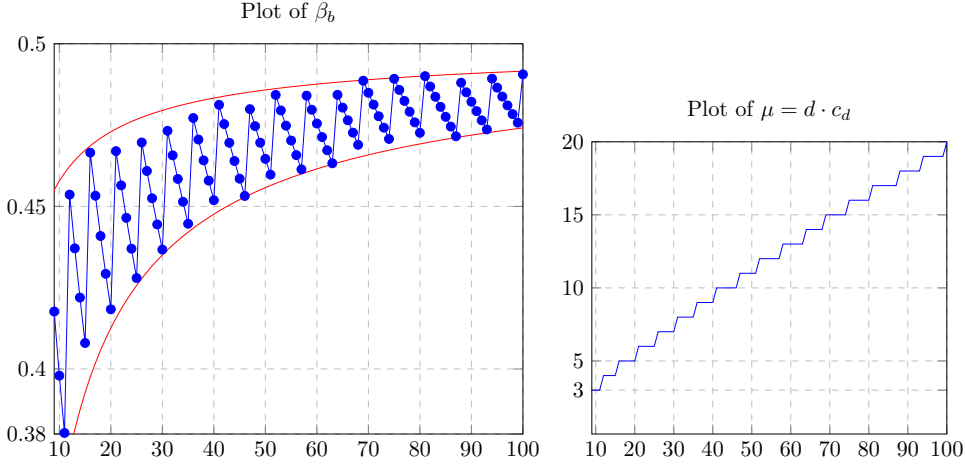
\end{lemma}

\begin{proof}[Proof of \cref{th:PoSLower_weighted}]
Fix some
integer $d\geq 9$.
Our lower bound instance consists of $n+\mu$ players and $n+\mu+1$ facilities, where
$\mu {:=} c\cdot d$ for $c=c_d$ defined as in~\eqref{eq:c_def}. In particular then,
due to \eqref{eq:c_def_lemma_2} of \cref{lemma:def_c}, $\mu\geq 3$ is an integer.
You can think of $n$ as a very large integer, since at the end we will take
$n\to\infty$. Every player $i=1,2,\dots,n+\mu$ has a weight of $w_i = w^i$, where $w
= 1+\frac{1}{\Phi_d}$.

It will be useful for subsequent computations to notice that
$$w^d=\left(1+\frac{1}{\Phi_d}\right)^d=\frac{(\Phi_d+1)^d}{\Phi_d^d}=\frac{\Phi_d^{d+1}}{\Phi_d^d}=\Phi_d,$$
$$w^{d+1}=w^d\cdot w=\Phi_d\left(1+\frac{1}{\Phi_d}\right)=\Phi_d+1.$$ Let us also
define
$$
\alpha=\alpha(\mu)
{:=}\sum_{j=1}^{\mu}w^{-j} =\frac{1-w^{-\mu}}{w-1} =\frac{1-(w^{d})^{-c}}{w-1}
=\frac{1-\Phi_d^{-c}}{1+\frac{1}{\Phi_d}-1} =\Phi_d\left(1-\Phi_d^{-c}\right) =\beta
\Phi_d,
$$
where $\beta=\beta_d$ is defined as in \eqref{eq:b_def}. In the following we will
make extensive use of the observation that
\begin{equation*}
w^{-\mu}=\left(w^d\right)^{-c}=\Phi_d^{-c}=1-\beta.
\end{equation*}
Furthermore, for every $i\geq \mu+1$
$$
\sum_{j=i-\mu}^{i-1}w_j=\sum_{j=1}^{\mu}w^{i-j}=\alpha\cdot w^i
\qquad\text{and}\qquad
\sum_{j=i-\mu}^{i}w_j=(\alpha+1)\cdot w^i,
$$ and
$$
\sum_{\ell=1}^\infty w^{-\ell}=\frac{1}{w-1}=\frac{1}{1+\frac{1}{\Phi_d}-1}=\Phi_d.
$$

The facilities have latency functions
\begin{align*} 
c_{j}(t) &= \Phi_d(1-\beta)(\alpha+1)^d, &&\text{if}\;\;j=1,\dots,\mu, \\ 
c_{j}(t) &= w^{-j(d+1)}t^d, && \text{if}\;\; j=\mu+1,\dots, \mu+n,\\ 
c_{n+\mu+1}(t) &= 0.
\end{align*}

Every player $i$ has two available strategies, $s_i^*$ and $\tilde
s_i$. Eventually we will show that the profile $\vecc s^*$ corresponds
to the optimal solution, while $\tilde{\vecc s}$ corresponds to the
{\em unique} Nash equilibrium of the game.  Informally, at the former
the player chooses to stay at her ``own'' $i$-th facility, while at
the latter she chooses to deviate and play the $\mu$ following
facilities $i+1,\dots,i+\mu$. However, special care shall be taken for
the boundary cases of the first $\mu$ and last $\mu$ players, so for
any player $i$ we formally define $S_i=\sset{s_i^*,\tilde s_i}$ where
$s_i^*=\ssets{i}$ and
$$
\tilde s_i=
\begin{cases}
\ssets{\mu+1,\dots,\mu+i}, & \text{if}\;\; i=1,\dots,\mu,\\
\ssets{i+1,\dots,i+\mu}, & \text{if}\;\; i=\mu+1,\dots,n,\\
\ssets{i+1,\dots,n+\mu+1}, & \text{if}\;\; i=n+1,\dots,n+\mu.
\end{cases}
$$
These two outcomes, $\vecc s^*$ and $\tilde{\vecc s}$, are shown in
\cref{fig:lower_bound_weighted}.
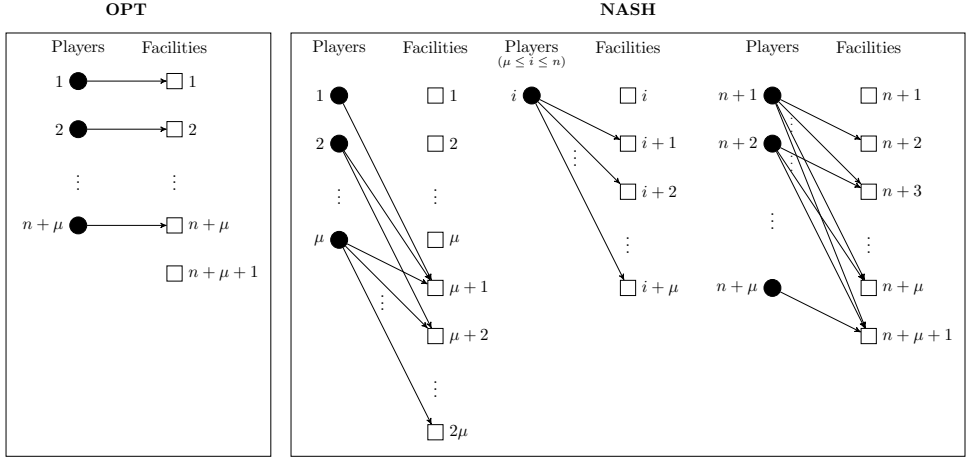
\begin{figure}
\centering
\scalebox{0.64}{
\begin{tikzpicture}[baseline=(current bounding box.north)]
\path[draw] (-1.5,1) rectangle (4,-7.8);
\node[font=\bfseries] at (1,1.5) {OPT};
\node at (0,0.7) {Players};
\node at (2,0.7) {Facilities};
\node[circle,draw,fill,label=left:$1$,minimum size=10pt] (p1) at (0,0) {};
\node[circle,draw,fill,label=left:$2$,minimum size=10pt] (p2) at (0,-1) {};
\node at (0,-2) {$\vdots$};
\node[circle,draw,fill,label=left:$n+\mu$,minimum size=10pt] (pn) at (0,-3) {};
\node[draw,label=right:$1$,minimum size=9pt] (e1) at (2,0) {};
\node[draw,label=right:$2$,minimum size=9pt] (e2) at (2,-1) {};
\node at (2,-2) {$\vdots$};
\node[draw,label=right:$n+\mu$,minimum size=9pt] (en) at (2,-3) {};
\node[draw,label=right:$n+\mu+1$,minimum size=9pt] (en1) at (2,-4) {};
\draw[->,>=stealth'] (p1) -- (e1);
\draw[->,>=stealth'] (p2) -- (e2);
\draw[->,>=stealth'] (pn) -- (en);
\end{tikzpicture}
}
\scalebox{0.64}{
\begin{tikzpicture}[baseline=(current bounding box.north)]
\path[draw] (-1,1.3) rectangle (13,-7.5);
\node at (0,1) {Players};
\node at (2,1) {Facilities};
\node[circle,draw,fill,label=left:$1$,minimum size=10pt] (p1) at (0,0) {};
\node[circle,draw,fill,label=left:$2$,minimum size=10pt] (p2) at (0,-1) {};
\node at (0,-2) {$\vdots$};
\node[circle,draw,fill,label=left:$\mu$,minimum size=10pt] (pmu) at (0,-3) {};
\node[draw,label=right:$1$,minimum size=9pt] (e1) at (2,0) {};
\node[draw,label=right:$2$,minimum size=9pt] (e2) at (2,-1) {};
\node at (2,-2) {$\vdots$};
\node[draw,label=right:$\mu$,minimum size=9pt] (emu) at (2,-3) {};
\node[draw,label=right:$\mu+1$,minimum size=9pt] (emu1) at (2,-4) {};
\node[draw,label=right:$\mu+2$,minimum size=9pt] (emu2) at (2,-5) {};
\node at (2,-6) {$\vdots$};
\node[draw,label=right:$2\mu$,minimum size=9pt] (e2mu) at (2,-7) {};
\draw[->,>=stealth'] (p1) -- (emu1);
\draw[->,>=stealth'] (p2) -- (emu1);
\draw[->,>=stealth'] (p2) -- (emu2);
\draw[->,>=stealth'] (pmu) -- (emu1);
\draw[->,>=stealth'] (pmu) -- (emu2);
\node at (0.9,-4.2) {$\vdots$};
\draw[->,>=stealth'] (pmu) -- (e2mu);
\node[font=\bfseries] at (1+4+1,-1.2+3) {NASH};
\node at (0+4,-2+3) {Players};
\node[font=\scriptsize] at (0+4,-2.3+3) {($\mu\leq i \leq n$)};
\node at (2+4,-2+3) {Facilities};
\node[circle,draw,fill,label=left:$i$,minimum size=10pt] (pi) at (0+4,-3+3) {};
\node[draw,label=right:$i$,minimum size=9pt] (ei) at (2+4,-3+3) {};
\node[draw,label=right:$i+1$,minimum size=9pt] (ei1) at (2+4,-4+3) {};
\node[draw,label=right:$i+2$,minimum size=9pt] (ei2) at (2+4,-5+3) {};
\node at (2+4,-6+3) {$\vdots$};
\node[draw,label=right:$i+\mu$,minimum size=9pt] (eimu) at (2+4,-7+3) {};
\draw[->,>=stealth'] (pi) -- (ei1);
\draw[->,>=stealth'] (pi) -- (ei2);
\node at (0.9+4,-4.2+3) {$\vdots$};
\draw[->,>=stealth'] (pi) -- (eimu);
\node at (0+9,1) {Players};
\node at (2+9,1) {Facilities};
\node[circle,draw,fill,label=left:$n+1$,minimum size=10pt] (pn1) at (0+9,0) {};
\node[circle,draw,fill,label=left:$n+2$,minimum size=10pt] (pn2) at (0+9,-1) {};
\node at (0+9,-2.5) {$\vdots$};
\node[circle,draw,fill,label=left:$n+\mu$,minimum size=10pt] (pnmu) at (0+9,-4) {};
\node[draw,label=right:$n+1$,minimum size=9pt] (en1) at (2+9,0) {};
\node[draw,label=right:$n+2$,minimum size=9pt] (en2) at (2+9,-1) {};
\node[draw,label=right:$n+3$,minimum size=9pt] (en3) at (2+9,-2) {};
\node at (2+9,-3) {$\vdots$};
\node[draw,label=right:$n+\mu$,minimum size=9pt] (enmu) at (2+9,-4) {};
\node[draw,label=right:$n+\mu+1$,minimum size=9pt] (enmu1) at (2+9,-5) {};
\draw[->,>=stealth'] (pn1) -- (en2);
\draw[->,>=stealth'] (pn1) -- (en3);
\node[font=\scriptsize] at (0.4+9,-0.5) {$\vdots$};
\draw[->,>=stealth'] (pn1) -- (enmu);
\draw[->,>=stealth'] (pn1) -- (enmu1);
\draw[->,>=stealth'] (pn2) -- (en3);
\node[font=\scriptsize] at (0.4+9,-1.3) {$\vdots$};
\draw[->,>=stealth'] (pn2) -- (enmu);
\draw[->,>=stealth'] (pn2) -- (enmu1);
\draw[->,>=stealth'] (pnmu) -- (enmu1);
\end{tikzpicture}
}
\caption{The social optimum $\vecc s^*$ and the unique Nash equilibrium
$\tilde{\vecc s}$ in the lower bound construction of \cref{th:PoSLower_weighted} for
general weighted congestion games.}
\label{fig:lower_bound_weighted}
\end{figure}

Notice here that any facility $j$ cannot get a load greater than the sum of the
weights of the previous $\mu$ players plus the weight of the $j$-th player. So, for
and any strategy profile $\vecc s$:
\begin{equation}
\label{eq:PoS_lower_load_bound_trivial} x_j(\vecc s) \leq \sum_{\ell=j-\mu}^{j}
w_{\ell}=(\alpha+1)w^j
\qquad\text{for all}\;\; j\geq \mu+1
\end{equation}

Next we will show that the strategy profile $\tilde{\vecc s}=(\tilde
s_1,\dots,\tilde s_{n+\mu})$ is the \emph{unique} Nash equilibrium of our congestion
game. We do that by proving that

\begin{enumerate}
\item It is a strongly \emph{dominant} strategy for any player
  $i=1,\dots,\mu$ to play $\tilde s_i$.

\item For any $i=\mu+1,\dots,n+\mu$, given that every player $k<i$ has chosen to play
$\tilde s_k$, then it is a strongly \emph{dominant} strategy for player $i$ to
deviate to $\tilde s_i$ as well.
\end{enumerate}

For the first condition, fix some player $i\leq  \mu$ and a strategy profile $\vecc
s_{-i}$ for the other players and observe that by choosing $\tilde s_i$, player $i$
incurs a cost of at most
\begin{align*}
C_i(\tilde s_i,\vecc s_{-i}) &=\sum_{j\in \tilde s_i}c_j(x_j(\tilde s_i))
\leq \sum_{\ell=\mu+1}^{\mu+i} c_{\ell}\left((\alpha+1) w^\ell\right)\\
&=
\sum_{\ell=\mu+1}^{\mu+i} w^{-\ell(d+1)}(\alpha+1)^d w^{\ell d} = (\alpha+1)^d
\sum_{\ell=\mu+1}^{\mu+i} w^{-\ell}\\ 
&< (\alpha+1)^d \cdot w^{-\mu}\cdot\sum_{\ell=1}^\infty
w^{-\ell} = (\alpha+1)^d \cdot (1-\beta)\cdot\Phi_d\\
 &= C_i(s_i^*,\vecc s_{-i}),
\end{align*} where in the first inequality we used the bound from
\eqref{eq:PoS_lower_load_bound_trivial}.

For the second condition, we will consider the deviations of the remaining players.
Fix now some $i=\mu+1,\dots,n$ and assume a strategy profile $\vecc s_{-i}=(\tilde
s_1,\dots, \tilde s_{i-1},\allowbreak s_{i+1},\allowbreak \dots,s_{n+\mu})$ for the remaining
players\footnote{For the remaining last $\mu$ players $i=n+1,\dots,n+\mu$ the proof
is similar, and as a matter of fact easier, since when these players
deviate to $\tilde s_i$ they also use the final ``dummy'' facility $n+\mu+1$ that
has zero cost.}. If player $i$ chooses strategy $s_i^*$ she will experience a cost
of
$$
C_i(s^*_i,\vecc s_{-i})=c_i\left(\sum_{\ell=i-\mu}^i
w_{\ell}\right)=c_i\left((\alpha+1)
w^i\right)=w^{-i(d+1)}(\alpha+1)^dw^{id}=(\alpha+1)^d w^{-i}.
$$
It remains to show that
\begin{equation}
\label{eq:Nash_condition_dev}
C_i(\tilde s_i,\vecc s_{-i}) < C_i(s^*_i,\vecc s_{-i})=(\alpha+1)^d
w^{-i}.
\end{equation}

The cost $C_i(\tilde s_i,\vecc s_{-i})$ is complicated to bound
immediately, for any profile $\vecc s_{-i}$. Instead, we will resort
to the following claim which characterizes the profile $\vecc s_{-i}$
where this cost is maximized, as shown
in~\cref{fig:dominate_format}.
Its proof can be found in \cref{app:PoS_lower_6_proof}.

\begin{claim}
\label{claim:best_s_i^*}
There exists a profile
$\vecc s'_{-i}$ with
\begin{enumerate}
\item $s_j'=s_j$ for all $j<i$ and $j>i+\mu$ \label{prop:dominating_prof_1}
\item $s_{i+\mu}'=s^*_{i+\mu}$ \label{prop:dominating_prof_2}
\item there exists some $k\in\sset{i+1,\dots,i+\mu-1}$ such that
\label{prop:dominating_prof_3}
$$
s_{j}'=\tilde s_j\qquad\text{for all}\;\;
j\in\sset{i+1,\dots,i+\mu-1}\setminus\ssets{k},
$$
\end{enumerate} that dominates $\vecc s_{-i}$, i.e.\
\begin{equation}
\label{eq:helper_PoS_lower_6}
C_i(\tilde s_i,\vecc s_{-i}) \leq C_i(\tilde s_i,\vecc s'_{-i}).
\end{equation}
\begin{figure}
\small
\centering
\scalebox{0.85}{
\begin{tikzpicture}[scale=1.5,nash/.style={circle,draw,fill=white,minimum
size=15pt,inner sep=0pt,path picture={\draw[->,very thick,>=stealth] (path picture
bounding box.west) -- (path picture bounding
box.east);}},opt/.style={circle,draw,fill=white,minimum size=15pt,inner sep=0pt,path
picture={\draw[->,very thick,>=stealth] (path picture bounding box.north) -- (path
picture bounding box.south);}},either/.style={circle,draw,fill=white,minimum
size=15pt,inner sep=0pt,path picture={\draw[->,very thick,>=stealth] (path picture
bounding box.west) -- (path picture bounding box.east);\draw[->,very
thick,>=stealth] (path picture bounding box.north) -- (path picture bounding
box.south);}}]
\draw[fill,gray!30] (1-0.2,0) rectangle (7+0.2,0.4);
\coordinate[nash,label=below:$i$] (i) at (0,0);
\coordinate[nash,label=below:$i+1$] (i+1) at (1,0);
\coordinate[nash,label=below:$k-1$] (k-1) at (2.5,0);
\coordinate[either,label=below:$k$] (k) at (3.5,0);
\coordinate[nash,label=below:$k+1$] (k+1) at (4.5,0);
\coordinate[nash,label=below:$i+\mu-1$] (i+n-1) at (6,0);
\coordinate[opt,label=below:$i+\mu$] (i+n) at (7,0);
\draw (i) -- (i+1);
\draw[dotted,thick] (i+1) -- (k-1);
\draw  (k-1) -- (k) -- (k+1);
\draw[dotted,thick] (k+1) -- (i+n-1);
\draw  (i+n-1) -- (i+n);
\draw[->,dotted,thick] (i+n) -- (8.5,0);
\draw[dotted,thick] (-1.5,0) -- (i);
\coordinate[nash,label=right: \quad{\sf NASH}] (nash) at (-0.5,2);
\coordinate[opt,label=right: \quad{\sf OPT}] (opt) at (2.5,2);
\coordinate[either,label=right: \quad{\sf NASH} or {\sf OPT}] (either) at (5.2,2);
\end{tikzpicture}}
\caption{The format of profile $\vecc s'_{-i}$ described in \cref{claim:best_s_i^*} and returned as output from Procedure
$\algoname{Dominate}(\vecc s_{-i},i)$ (see \cref{app:PoS_lower_6_proof}). All players $i+1,\dots,i+\mu$ (i.e., those who lie
within the window of interest of player $i$, depicted in grey) play according to the Nash equilibrium
$\tilde{\vecc s}$, except the last player $i+\mu$ (that plays according to the
optimal profile $\vecc s^*$) and at most one other $k$ (that may play either
$\tilde s_k$ or $s^*_k$).}
\label{fig:dominate_format}
\end{figure}
\end{claim}

By use of \cref{claim:best_s_i^*}, it remains to show
\begin{equation}
\label{eq:helper_PoS_lower_8}
 C_i(\tilde s_i,\vecc s_{-i}') < (\alpha+1)^d w^{-i},
\end{equation} just for the special case of profiles $\vecc s'$ that are described in \cref{claim:best_s_i^*} and also shown in \cref{fig:dominate_format}.
We do this in \cref{app:PoS_lower_8_proof}.

Summarizing, we proved that indeed $\tilde{\vecc s}$ is the \emph{unique} Nash
equilibrium of our congestion game. Finally, to conclude with lower-bounding the
Price of Stability, let us compute the social cost on profiles $\tilde{\vecc s}$ and
$\vecc s^*$. On $\vecc s^*$, any facility $j$ (except the last one) gets a load
equal to the weight of player $j$, so
\begin{align*} 
C(\vecc s^*) &=\sum_{j=1}^{n+\mu}w_j c_j(w_j)\\ 
&=\sum_{j=1}^{\mu}w^j
\Phi_d(1-\beta)(\alpha+1)^d + \sum_{j=\mu+1}^{n+\mu}w^{j}w^{-j(d+1)}(w^j)^{d}\\
&=\Phi_d(1-\beta)(\alpha+1)^d\sum_{j=1}^\mu w^j+\sum_{j=\mu+1}^{\mu+n} 1\\ 
&=\Phi_d(1-\beta)(\alpha+1)^d
w\frac{w^{\mu}-1}{w-1}+n\\ 
&=\Phi_d(1-\beta)(\alpha+1)^d
\left(1+\frac{1}{\Phi_d}\right)\frac{\frac{1}{1-\beta}-1}{1+\frac{1}{\Phi_d}-1}+n\\
&=\Phi_d(1-\beta)(\alpha+1)^d (\Phi_d+1)\frac{\beta}{1-\beta} +n\\ 
&\leq n+
\Phi_d(\Phi_d+1)\beta(\alpha+1)^{d}.
\end{align*}
On the other hand, at the unique Nash equilibrium $\tilde{\vecc s}$
each facility $j\geq \mu+1$ receives a load equal to the sum of the weights of the
previous $\mu$ players, i.e.
$$
x_j(\tilde{\vecc s})=\sum_{\ell=j-\mu}^{j-1} w_{\ell}=\alpha w^j
$$
so
$$
C(\tilde{\vecc s})
\geq\sum_{j=\mu+1}^{n+\mu}x_j(\tilde{\vecc s})c_j(x_j(\tilde{\vecc s}))
=\sum_{j=\mu+1}^{n+\mu} w^{-j(d+1)}\left(\alpha w^j\right)^{d+1}
=\alpha^{d+1}\sum_{j=\mu+1}^{\mu+n}1=\alpha^{d+1} n.
$$
By taking $n$ arbitrarily large we get a lower bound on the Price of Stability of
$$
\lim_{n\to\infty}\frac{C(\tilde{\vecc s})}{C(\vecc s^*) }
\geq 
\lim_{n\to\infty}\frac{\alpha^{d+1} n}{n+\Phi_d(\Phi_d+1)\beta(\alpha+1)^{d}} 
=\alpha^{d+1}=(\beta\Phi_d)^{d+1},
$$
where from \cref{lemma:def_c} we know that $\frac{1}{3}\leq \beta=\frac{1}{2}-o(1)$.
\end{proof}

\subsubsection{Network Games}
\label{sec:network}
Due to the rather simple structure of the players' strategy sets in the lower bound
construction of \cref{th:PoSLower_weighted}, it can be readily extended to network
games as well:
\begin{proposition}
\label{cor:network_lower}
\cref{th:PoSLower_weighted} applies also to \emph{network} weighted congestion games.
\end{proposition}
\begin{proof}
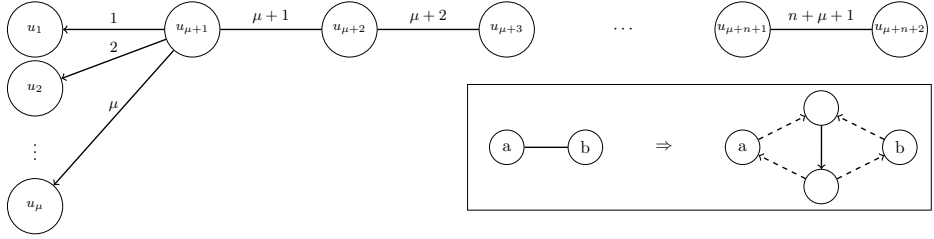
\begin{figure}
\centering
\tikzstyle{vertex}=[circle,minimum size=20pt,inner sep=0pt]
\tikzstyle{gadgetvertex}=[circle,draw,minimum size=20pt,inner sep=0pt]
\tikzstyle{edge} = [draw,thick,-]
\tikzstyle{weight} = [above,font=\small]
\usetikzlibrary{arrows, petri, topaths, graphs, graphs.standard}
\tikzset{vertex/.style = {shape=circle, draw, minimum size=32pt,inner sep=0pt,font=\footnotesize}}
\resizebox*{.95\textwidth}{!}{%
\begin{tikzpicture}[scale=1.6, auto,swap]
        \node[vertex] (A) at (0,0) {$u_1$};
        \node[vertex] (B) at (0,-0.75) {$u_2$};
        \node[vertex] (C) at (0,-2.25) {$u_\mu$};

        \node[vertex] (D) at (2,0) {$u_{\mu+1}$};
        \node[vertex] (E) at (4,0) {$u_{\mu+2}$};
        \node[vertex] (F) at (6,0) {$u_{\mu+3}$};
        \node[vertex] (G) at (9,0) {$u_{\mu+n+1}$};
        \node[vertex] (H) at (11,0) {$u_{\mu+n+2}$};

        \path[edge,<-] (A) -- node[weight]{$1$} (D);
        \path[edge,<-] (B) -- node[weight]{$2$} (D);
        \path[edge,<-] (C) -- node[weight]{$\mu$} (D);
        \path[edge] (D) -- node[weight]{$\mu+1$} (E);
        \path[edge] (E) -- node[weight]{$\mu+2$} (F);
        \path[edge] (G) -- node[weight]{$n+\mu+1$} (H);

        \node at (7.5,0) {$\cdots$};
        \node at (0,-1.5) {$\vdots$};

        \node[gadgetvertex] (L1) at (4+2,-1.5) {a};
        \node[gadgetvertex] (L2) at (5+2,-1.5) {b};
        \path[edge] (L1) --  (L2);
        \node at (6+2,-1.5) {$\Rightarrow$};
         \node[gadgetvertex] (R1) at (7+2,-1.5) {a};
        \node[gadgetvertex] (R4) at (9+2,-1.5) {b};
         \node[gadgetvertex] (R2) at (8+2,-1.0) {$ $};
        \node[gadgetvertex] (R3) at (8+2,-2) {$ $};
        \path[draw,thick,dashed,->] (R1) -- (R2);
        \path[draw,thick,dashed,->] (R3) -- (R1);
        \path[draw,thick,dashed,->] (R4) -- (R2);
        \path[draw,thick,dashed,->] (R3) -- (R4);
       \path[draw,thick,->] (R2) -- (R3);
       \path[draw] (L1)+(-0.5,+0.8) rectangle (9.4+2,-2.3);
\end{tikzpicture}
}
\caption{Transformation of the lower bound instance of \cref{th:PoSLower_weighted}
for general weighted congestion games to a network game, as described in
\cref{cor:network_lower}.}
\label{fig:ncg}
\end{figure}
We arrange the resources from the instance in the proof of
\cref{th:PoSLower_weighted} as edges in a graph as depicted in \cref{fig:ncg}. In
particular, for all $j=\mu+1,\dots,n+\mu+1$, resource $j$ from
\cref{th:PoSLower_weighted} corresponds to edge $(u_{j},u_{j+1})$. For the special
case of the first $\mu$ resources, for $j=1,\dots,\mu$, we represent resource $j$ by
the directed edge $(u_{\mu+1},u_j)$.

Regarding strategies, recall from the instance used in the proof of
\cref{th:PoSLower_weighted} that, for each $i=\mu,\dots,n$, player $i$ has two
available strategies: resource $\sset{i}$ or resources $\sset{i+1,\dots,i+\mu}$. To
map this to our network instance, we set the source node $o_i$ of player $i$ to be
$o_i=u_{i+1}$ and introduce a new destination node $t_i$ connected to the rest of
the graph by zero-latency directed edges $(u_{i},t_{i})$ and $(u_{i+\mu+1},t_{i})$.
In that way, strategy $\sset{i}$ of \cref{th:PoSLower_weighted} corresponds to the
path $o_i=u_{i+1}\to u_{i} \to t_i$ of our graph, while strategy
$\sset{i+1,\dots,i+\mu}$ to path $o_i=u_{i+1}\to u_{i+2}\to \dots \to u_{i+\mu+1}
\to t_{i}$. To avoid clutter, these destination nodes $t_i$ and the corresponding
zero-latency edges are not depicted in \cref{fig:ncg}. In an analogous way, we can
set the sources and destinations of the remaining first $i=1,\dots,\mu-1$ and last
$i=n+1,\dots,n+\mu+1$ players, taking into consideration their specially restricted
strategy sets in the construction of the proof of~\cref{th:PoSLower_weighted}.

Summarizing, each player $i\in[1,n+\mu]$ has to route its traffic from $o_i$
to $t_i$, where $$ o_i=
\begin{cases}
u_{\mu+1}, & \text{if}\;\; i=1,\dots,\mu,\\
u_{i+1}, & \text{if}\;\; i=\mu+1,\dots,n+\mu,
\end{cases}
$$
and nodes $t_i$ are connected with \emph{zero latency} edges as follows:
\begin{itemize}
\item For each $i\in[1,n+\mu]$ there is a directed zero cost edge from $u_i$ to $t_i$.
\item For each $i\in[1,n]$ there is a directed zero cost edge from $u_{\mu+1+i}$ to $t_i$.
\item For each $i\in[n+1,n+\mu]$ there is a directed zero cost edge from $u_{\mu+n+2}$ to $t_i$.
\end{itemize}
Then, by construction, each player $i$ has two available $o_i\to t_i$ paths, which
correspond directly to strategy sets $s_i^*$ and $\tilde{s_i}$ used in the proof of
\cref{th:PoSLower_weighted}.

There is one issue left to complete our network game construction: we have not yet
set a \emph{direction} on some of our edges, namely $(u_{i},u_{i+1})$ for
$i=\mu+1,\dots,n+\mu+1$ which are depicted also as undirected edges in
\cref{fig:ncg}. This is due to the fact that, by our construction so far, these
edges can be used in \emph{both} directions: by player $i$ (left-to-right) or player
$i+1$ (right-to-left). Thus, to turn our instance to a valid \emph{directed}
network, we need to replace such edges with a ``gadget'' that essentially forces
both players, no matter from which direction they enter the edge, to use it in the
same direction and \emph{both} contribute to its load. This can be achieved by using
the structure depicted in the bottom right corner of \cref{fig:ncg}.
\end{proof}

\subsection{Singleton Games}
\label{sec:lower_singleton}
In this section we give an exponential lower bound for {\em singleton}
weighted congestion games with polynomial latency functions.
The following theorem handles also approximate equilibria and provides
a lower bound on the Price of Stability in a very strong sense; even
if one allows for the best approximate equilibrium with approximation
factor $\alpha=o\left(\frac{d}{\ln d}\right)$, then its cost is
lower-bounded by $\omega(\text{poly}(d))$ times the optimal cost.\footnote{\label{foot:1} To see
this, just take any upper bound of $\frac{d+1}{c\ln (d+1)}$ on
$\alpha$, for a constant $c>2$. Then, the lower bound in
\eqref{eq:Pos_approx_lower_parallel_1} becomes $\Omega(d^{c-1})$.}
In other words, in order to achieve polynomial (with respect to $d$) guarantees on the Price of
Stability, one has to consider $\varOmega\left(\frac{d}{\ln d}\right)$-approximate
equilibria---almost linear in $d$; this shows that our positive result in
\cref{th:PoS_upper_general_exact_only}, of the following
\cref{sec:upper_bound_therem}, is almost tight. This is furthermore
complemented by \cref{th:PoS_OPT}, where we show that the socially
optimum profile is a $(d+1)$-approximate equilibrium (achieving an
optimal Price of Stability of $1$).

\begin{theorem}
\label{th:PoS_lower_approx}
For any positive integer $d$ and any real $\alpha\in [1,d)$, the
$\alpha$-approximate (mixed) Price of Stability of weighted (singleton) congestion
games with polynomial latencies of degree at most $d$ is at least
\begin{equation}
\frac{1}{e(d+1)}\left(1+\frac{1}{\alpha}\right)^{d+1}.
\label{eq:Pos_approx_lower_parallel_1}
\end{equation}
In particular, for the special case of $\alpha=1$, we derive that the Price of Stability of \emph{exact} equilibria is $\varOmega(2^{d}/d)$.
\end{theorem}
\begin{proof}
Fix a positive integer $d$ and the desired approximation parameter
$\alpha\in[1,d)$. Also, let $\gamma\in(\alpha,d)$ be a parameter
arbitrarily close to $\alpha$.  Our instance consists of $n$ players
with weights $w_i=w^i$, $i=1,2,\dots,n$, where we set
\begin{equation}
\label{eq:lower_singleton_weight}
w=\gamma\frac{d+1}{d-\gamma}>\gamma,
\end{equation}
the inequality holding due to the fact that $d+1>d-\gamma>0$.
At the end of our construction we will take $n\to\infty$, so one can
think of $n$ as a very large integer.
There are $n+1$ facilities with latency
functions
\begin{align*}
c_{1}(t) &= \gamma w^d(w+1)^d, \\
c_{j}(t) &= (\gamma w^d)^{2-j}\cdot t^d, && j=2,\dots, n,\\
c_{{n+1}}(t) &= \gamma^{1-n}w^d(w+1)^d.
\end{align*}

Any player $i$ has exactly two strategies, $s^*_i=\{i\}$ and $\tilde s_i=\{i+1\}$
i.e., $S_i=\sset{\ssets{i},\ssets{i+1}}$ for all $i=1,\dots,n$.  Let $\vecc s^*,
\tilde{\vecc s}$ be the strategy profiles where every player $i$ plays $s_i^*,
\tilde s_i$ respectively. These two outcomes, $\vecc s^*$ and $\tilde{\vecc s}$ are
depicted in \cref{fig:lower_bound_singleton}.
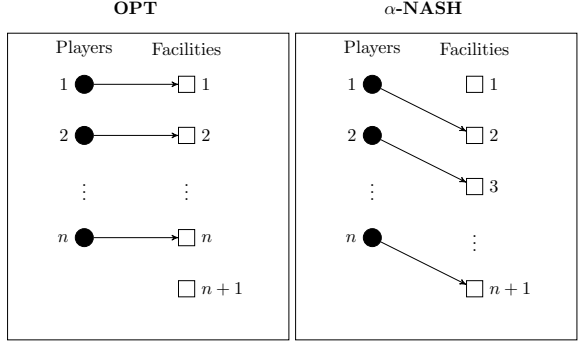
\begin{figure}
\centering
\scalebox{0.68}{
\begin{tikzpicture}[baseline=(current bounding box.north)]
\path[draw] (-1.5,1) rectangle (4,-5);
\node[font=\bfseries] at (1,1.5) {OPT};
\node at (0,0.7) {Players};
\node at (2,0.7) {Facilities};
\node[circle,draw,fill,label=left:$1$,minimum size=10pt] (p1) at (0,0) {};
\node[circle,draw,fill,label=left:$2$,minimum size=10pt] (p2) at (0,-1) {};
\node at (0,-2) {$\vdots$};
\node[circle,draw,fill,label=left:$n$,minimum size=10pt] (pn) at (0,-3) {};
\node[draw,label=right:$1$,minimum size=9pt] (e1) at (2,0) {};
\node[draw,label=right:$2$,minimum size=9pt] (e2) at (2,-1) {};
\node at (2,-2) {$\vdots$};
\node[draw,label=right:$n$,minimum size=9pt] (en) at (2,-3) {};
\node[draw,label=right:$n+1$,minimum size=9pt] (en1) at (2,-4) {};
\draw[->,>=stealth'] (p1) -- (e1);
\draw[->,>=stealth'] (p2) -- (e2);
\draw[->,>=stealth'] (pn) -- (en);
\end{tikzpicture}
\begin{tikzpicture}[baseline=(current bounding box.north)]
\path[draw] (-1.5,1) rectangle (4,-5);
\node[font=\bfseries] at (1,1.5) {$\alpha$-NASH};
\node at (0,0.7) {Players};
\node at (2,0.7) {Facilities};
\node[circle,draw,fill,label=left:$1$,minimum size=10pt] (p1) at (0,0) {};
\node[circle,draw,fill,label=left:$2$,minimum size=10pt] (p2) at (0,-1) {};
\node at (0,-2) {$\vdots$};
\node[circle,draw,fill,label=left:$n$,minimum size=10pt] (pn) at (0,-3) {};
\node[draw,label=right:$1$,minimum size=9pt] (e1) at (2,0) {};
\node[draw,label=right:$2$,minimum size=9pt] (e2) at (2,-1) {};
\node[draw,label=right:$3$,minimum size=9pt] (e3) at (2,-2) {};
\node at (2,-3) {$\vdots$};
\node[draw,label=right:$n+1$,minimum size=9pt] (en1) at (2,-4) {};
\draw[->,>=stealth'] (p1) -- (e2);
\draw[->,>=stealth'] (p2) -- (e3);
\draw[->,>=stealth'] (pn) -- (en1);
\end{tikzpicture}
}
\caption{The social optimum $\vecc s^*$ and the unique $\alpha$-approximate equilibrium $\tilde{\vecc s}$ in the lower bound construction of \cref{th:PoS_lower_approx} for singleton weighted congestion games.}
\label{fig:lower_bound_singleton}
\end{figure}
One should think of
$\vecc s^*$ as the socially optimal profile. We will show that
$\tilde{\vecc s}$ is the \emph{unique} $\alpha$-approximate Nash
equilibrium of our game. To ensure this, it suffices to require the
following, which corresponds to eliminating all other possible
strictly dominated $\alpha$-approximate equilibria:

\begin{enumerate}
\item It is a strictly $\alpha$-dominant strategy for player $1$ to
  use facility $2$, i.e. 
  $$\alpha C_1(\tilde s_1,\vecc s_{-i})<
  C_1(\vecc s)$$ 
  for any profile $\vecc s$.
\item For any $i=2,\dots,n$, if every player $k<i$ has chosen facility $k+1$ then it is a strictly  $\alpha$-dominant strategy for player $i$ to chose facility $i+1$, i.e.
$$
\alpha C_i(\tilde s_1,\dots,\tilde s_{i-1},\tilde s_i, s_{i+1},\dots,s_n) < C_i(\tilde s_1,\dots,\tilde s_{i-1},s_i,s_{i+1},\dots,s_n) $$
 for any strategies $(s_i,s_{i+1},\dots,s_n)\in S_i\times\dots\times S_n$.
\end{enumerate}

For the first condition, since facility $2$ can be used by at most
players $1$ and $2$, and $\gamma>\alpha$, it is enough to show that $\gamma
c_{2}(w_1+w_2)\leq c_{1}(w_1)$. Indeed
$$
\gamma c_{2}(w_1+w_2)=\gamma (\gamma w^d)^{2-2}(w+w^2)^d=\gamma w^d(1+w)^d=c_{1}(w_1).
$$

Similarly, for the second condition, it suffices to show that $
\gamma c_{i+1}(w_i+w_{i+1}) \leq c_{i}(w_{i-1}+w_i) $ for
$i=2,\dots,n-1$, and $ \gamma c_{n+1}(w_n)\leq
c_{n}(w_{n-1}+w_n) $ for the special case of $i=n$. This is
because, facility $i+1$ can be used by at most players $i$ and
$i+1$, while facility $i$ is already being used by player
$i-1$. Indeed, for any $i=2,\dots,n-1$ we see that:
\begin{align*}
\gamma c_{i+1}(w_i+w_{i+1}) &=\gamma (\gamma w^d)^{2-(i+1)}(w^i+w^{i+1})^d
= (\gamma w^d)^{2-i}(w^{i-1}+w^i)^d\\
&= c_{i}(w_{i-1}+w_i),
\end{align*}
while for $i=n$,
\begin{align*}
c_{n}(w_{n-1}+w_n) &= (\gamma w^d)^{2-n}(w^{n-1}+w^n)^d\\
  &= \gamma^{2-n}w^{d(2-n)+d(n-1)}(w+1)^d\\
  &=\gamma\cdot \gamma^{1-n}w^d(w+1)^d\\
  &=\gamma c_{n+1}(w_n).
\end{align*}

The social cost at equilibrium $\tilde{\vecc s}$ is at least the cost of player $n$ at $\tilde{\vecc s}$, that is,
$$
C(\tilde{\vecc s})\geq w_n c_{n+1}(w_n)=w^n\cdot \gamma^{1-n}w^d(1+w)^d=\left(\frac{w}{\gamma}\right)^n \gamma \cdot w^d(1+w)^d
$$
On the other hand, consider the strategy profile $\vecc s^*$ where every player $i$ chooses facility $i$:
\begin{align*}
C(\vecc s^*)
 &=w_1 c_{1}(w_1)+\sum_{i=2}^nw_i c_{i}(w_i)\\
 &= \gamma w^{d+1}(1+w)^d+ \sum_{i=2}^nw^i (\gamma w^d)^{2-i} w^{id}\\
 &=\gamma w^{d+1}(1+w)^d+\gamma^2w^{2d}\sum_{i=2}^n\left(\frac{w}{\gamma}\right)^i\\
 &=\gamma w^{d+1}(1+w)^d+\gamma^2w^{2d}\cdot \left(\frac{w}{\gamma}\right)^2\frac{\left(\frac{w}{\gamma}\right)^{n-1}-1}{\frac{w}{\gamma}-1}\\
 &\leq\gamma w^{d+1}(1+w)^d+\gamma^2w^{2d}\cdot \left(\frac{w}{\gamma}\right)^2\frac{\left(\frac{w}{\gamma}\right)^{n-1}}{\frac{w}{\gamma}-1}\\
 &= \left(\frac{w}{\gamma}\right)^{n}\gamma \cdot
 \left[\left(\frac{w}{\gamma}\right)^{-n}\cdot w^{d+1}(w+1)^d + \frac{w^{2d+1}}{\frac{w}{\gamma}-1}\right]
\end{align*}

Recall now that, from \eqref{eq:lower_singleton_weight}, $\frac{w}{\gamma}>1$, and thus $\lim_{n\to\infty}\left(\frac{w}{\gamma}\right)^{-n}=0$. So, as the number of players $n$ grows large we get the following lower bound on the Price of Stability:
$$
\lim_{n\to\infty}\frac{C(\tilde{\vecc s})}{C(\vecc s^*)}
\geq \lim_{n\to\infty}\frac{w^d(1+w)^d}{\left(\frac{w}{\gamma}\right)^{-n}\cdot w^{d+1}(w+1)^d + \frac{w^{2d+1}}{\frac{w}{\gamma}-1}}=\left(\frac{w}{\gamma}-1\right)\frac{(1+w)^d}{w^{d+1}}.
$$
Since $\gamma$ is chosen arbitrarily close to $\alpha$, deploying
\eqref{eq:lower_singleton_weight} to substitute $w$, the above lower bound can be
written as
\begin{align*}
\lim_{n\to\infty}\frac{C(\tilde{\vecc s})}{C(\vecc s^*)}
& \geq
\left(\frac{d+1}{d-\alpha}-1\right)\left[1+\frac{\alpha(d+1)}{d-\alpha}\right]^d/\left[\frac{\alpha(d+1)}{d-\alpha}\right]^{d+1} \\
&= \left(\frac{\alpha+1}{d-\alpha} \right)\left[\frac{d(\alpha+1)}{d-\alpha}\right]^d/\left[\frac{\alpha(d+1)}{d-\alpha}\right]^{d+1}\\
&= \frac{(\alpha+1)^{d+1} d^d}{\alpha^{d+1} (d+1)^{d+1}}\\
&= \frac{1}{d+1}\left(1-\frac{1}{d+1}\right)^d\left(1+\frac{1}{\alpha}\right)^{d+1} \\
&\geq \frac{1}{e}\frac{1}{d+1}\left(1+\frac{1}{\alpha}\right)^{d+1}.
\end{align*}
\end{proof}

\section{Upper Bounds}
\label{sec:upper_bounds}

The negative results of the previous sections involve constructions
where the ratio $W$ of the largest to smallest weight can be
exponential in $d$. In the main theorem (\cref{th:PoS_upper_general})
of this section we present an analysis which is sensitive to this
parameter $W$, and identify conditions under which the performance of
approximate equilibria can be significantly improved.

Our upper bound approach is based on the design of a suitable
approximate potential function and has three main steps. First, in
\cref{sec:potential_general_method}, we set up a framework for the
definition of this function by identifying conditions that, on the one
hand, certify the existence of an approximate equilibrium and, on the other,
provide guarantees about its efficiency.
Then, in \cref{sec:faulhaber}, by use of the Euler-Maclaurin summation formula we present a
general form of an approximate potential function, which extends
Rosenthal's potential for weighted congestion games (see also \cref{sec:EulerMaclaurin}).
Finally, in \cref{sec:upper_bound_therem}, we deploy this potential for polynomial latencies.  Due to its analytic description, our
potential differs from other extensions of the Rosenthal's potential
that have appeared in previous work, and we believe that this contribution might be
of independent interest, and applied to other classes of
latency functions.

\subsection{The Potential Method}
\label{sec:potential_general_method}

In the next lemma we lay the ground for the design and analysis of
approximate potential functions, by supplying conditions that not only
provide guarantees for the existence of approximate equilibria, but also for their performance with respect to the social optimum.
In the premises of the lemma, we
give conditions on the resource functions $\phi_e$, having in
mind that $\varPhi(\vecc s)=\sum_{e\in E}\phi_e(x_e(\vecc s))$ will eventually
serve as the ``approximate'' potential function.

\begin{lemma}
  \label{lemma:approximate_PoS_ratios}
  Consider a weighted congestion game with latency functions $c_e$,
  for each facility $e\in E$, and player weights $w_i$, for each
  player $i\in N$. If there exist functions $\phi_e:\R_{\geq 0}\map
  \R$ and parameters $\alpha_1,\alpha_2,\beta_1,\beta_2>0$ such that
  for any facility $e$ and player weight $w\in\sset{w_1,\dots,w_n}$
\begin{equation}
\label{eq:cond_approx_eq}
\alpha_1\leq \frac{\phi_e(x+w)-\phi_e(x)}{w\cdot c_e(x+w)} \leq \alpha_2,
\qquad\text{for all}\;\; x\geq 0,
\end{equation}
and
\begin{equation}
\label{eq:cond_approx_cost}
  \beta_1\leq \frac{\phi_e(x)}{x\cdot c_e(x)} \leq \beta_2,
 \qquad\text{for all}\;\; x\geq \min_n w_n,
\end{equation}
then our game has an $\frac{\alpha_2}{\alpha_1}$-approximate pure Nash
equilibrium which, furthermore, has Price of Stability at most
$\frac{\beta_2}{\beta_1}$.
\end{lemma}
\begin{proof}
  Denote $\alpha=\frac{\alpha_2}{\alpha_1}$,
  $\beta=\frac{\beta_2}{\beta_1}$. First we will show that the
  function $\varPhi(\vecc s)=\sum_{e\in E}\phi_e(x_e(\vecc s))$
  (defined over all feasible outcomes $\vecc s$) is an
  $\alpha$-approximate potential, i.e.\ for any profile $\vecc s$, any
  player $i$ and strategy $s_i'\in S_i$,
$$
C_i(s_i',\vecc s_{-i})< \frac{1}{\alpha}C_i(\vecc s)
\quad \then \quad
\varPhi(s_i',\vecc s_{-i}) < \varPhi(\vecc s).
$$
This would be sufficient to establish the existence of a pure $\alpha$-approximate equilibrium, since any (local) minimizer of $\varPhi$ will do. So, it is enough to prove that
$$
\varPhi(s_i',\vecc s_{-i})-\varPhi(\vecc s) \leq w_i\alpha_1\left[\alpha \cdot C_i(s_i',\vecc s_{-i}) - C_i(\vecc s) \right].
$$
Indeed, if for simplicity we denote $x_e=x_e(\vecc s)$ and
$x_e'=x_e(s_i',\vecc s_{-i})$ for all facilities $e$, we can compute
\begin{align*}
\varPhi(s_i',\vecc s_{-i})-\varPhi(\vecc s)
&= \sum_{e\in E} \left[ \phi_e(x_e')-\phi_e(x_e)\right]\\
&= \sum_{e\in s_i'\setminus s_i}\left[ \phi_e(x_e+w_i)-\phi_e(x_e)\right] + \sum_{e\in s_i\setminus s_i'}\left[ \phi_e(x_e-w_i)-\phi_e(x_e)\right]\\
&\leq \alpha_2\sum_{e\in s_i'\setminus s_i} w_i c_e(x_e+w_i) - \alpha_1\sum_{e\in s_i\setminus s_i'} w_i c_e(x_e)\\
&\leq  w_i\alpha_2\left(\sum_{e\in s_i'\setminus s_i} c_e(x_e+w_i) +\sum_{e\in s_i'\inters s_i} c_e(x_e)\right)\\ 
&\qquad\qquad\qquad- w_i\alpha_1 \left(\sum_{e\in s_i\setminus s_i'} c_e(x_e) +\sum_{e\in s_i'\inters s_i} c_e(x_e)\right)\\
&= w_i\alpha_1\left[\alpha\left(\sum_{e\in s_i'\setminus s_i}  c_e(x_e+w_i) +\sum_{e\in s_i'\inters s_i} c_e(x_e)\right)\right.\\ 
&\qquad\qquad\qquad\qquad-\left. \left(\sum_{e\in s_i\setminus s_i'} c_e(x_e) +\sum_{e\in s_i'\inters s_i} c_e(x_e)\right) \right]\\
&= w_i\alpha_1\left[\alpha C_i(s_i',\vecc s_{-i}) - C_i(\vecc s) \right].
\end{align*}
where the first inequality holds due to~\eqref{eq:cond_approx_eq}
and the second one because $\alpha_2\geq \alpha_1$.

Next, for the upper bound of $\beta$ on the Price of Stability, it is enough to show that for any profiles $\vecc s$, $\vecc s'$,
$$
\varPhi(\vecc s) \leq \varPhi(\vecc s') \quad \then \quad C(\vecc s) \leq \beta \cdot C(\vecc s'),
$$
because then, if $\vecc s^*\in\argmin_{\vecc s} C(\vecc s)$ is an optimal-cost
profile and $\tilde{\vecc s}\in\argmin_{\vecc s} \varPhi(\vecc s)$ is a
\emph{global} minimizer of $\varPhi$, then $C(\tilde{\vecc s})\leq \beta C(\vecc
s^*)$ (and furthermore, as a minimizer of $\varPhi$, $\tilde{\vecc s}$ is clearly an
$\alpha$-approximate equilibrium as well; see the first part of the current proof).
Indeed, denoting $x_e=x_e(\vecc s)$, $x_e'=x_e(\vecc s')$ for simplicity, we have:
\begin{align*}
\varPhi(\vecc s')-\varPhi(\vecc s) & = \sum_{e\in E}\phi_e(x_e')-\sum_{e\in E}\phi_e(x_e)\\
&\leq \beta_2\sum_{e\in E}x_e'c_e(x_e') -\beta_1\sum_{e\in E}x_ec_e(x_e)\\
&=\beta_2 C(\vecc s') - \beta_1 C(\vecc s)\\
&=\beta_1\left(\beta C(\vecc s') -  C(\vecc s) \right),
\end{align*}
where for the first inequality we deployed~\eqref{eq:cond_approx_cost}.
\end{proof}

\subsection{Faulhaber's Potential}
\label{sec:faulhaber}
In this section we propose an approximate potential function, which is based on the
following classic number-theoretic result, known as Faulhaber's
formula\footnote{See, e.g., \citep[p.~287]{Knuth1993} or
\citep[p.~106]{Conway1996}). Johann Faulhaber~\citep{faulhaber1631} was the first to
discover the formula and express it in a systematic way, up to the power of $m=17$.
Jakob Bernoulli was able to state it in its full generality as his famous
\emph{Summæ Potestatum}~\citep[p.~97]{bernoulli1713}, by introducing what are now
known as \emph{Bernoulli numbers} (see also \cref{foot:bernoulli_nums}). The first
to rigorously prove the formula was Carl Jacobi~\citep{jacobi1834}.}, which states
that for any positive integers $n,m$,
\begin{align}
\sum_{k=1}^n k^m
  &= \frac{1}{m+1}\sum_{j=0}^m(-1)^j\binom{m+1}{j} B_jn^{m+1-j}\notag \\
  &= \frac{1}{m+1}n^{m+1}+\frac{1}{2}n^m+\frac{1}{m+1}\sum_{j=2}^m\binom{m+1}{j} B_jn^{m+1-j} ,\label{eq:faulhaber_formula_expand}
\end{align}
where the coefficients $B_j$ are the usual Bernoulli
numbers.\footnote{\label{foot:bernoulli_nums}See, e.g., \citep[Chapter 6.5]{Graham1989a} or
  \citep[Chapter 23]{Abramowitz1964}. The first Bernoulli numbers are:
  $B_0=1, B_1=-1/2, B_2=1/6, B_3=0, B_4=-1/30, \dots$. Also, we know
  that $B_{j}=0$ for all \emph{odd} integers $j\geq 3$.} In
particular, this shows that the sum of the first $n$ powers with exponent $m$ can
be expressed as a polynomial of $n$ with degree $m+1$.  Furthermore,
this sum corresponds to the well-known potential of Rosenthal~\citep{Rosenthal1973a} for
\emph{unweighted} congestion games when the latency function is the
monomial $x\mapsto x^m$.

Based on the above observation, we go beyond just integer values of
$n$, and generalize this idea to all positive reals; in that way, we
design a ``potential'' function that can handle different player
weights and, furthermore, incorporate in a more powerful, analytically
smooth way, approximation factors with respect to both the Price of
Stability, as well as the approximation parameter of the equilibrium
(in the spirit of \cref{lemma:approximate_PoS_ratios}). A natural way
to do that is to directly generalize
\eqref{eq:faulhaber_formula_expand} and simply define, for any real
$x\geq 0$ and positive integer $m$,
\begin{equation}
\label{eq:potential_mono_def}
S_m(x) {:=} \frac{1}{m+1}x^{m+1}+\frac{1}{2}x^m,
\end{equation}
keeping just the first two significant terms.\footnote{See
  \cref{sec:full_faulhaber_note} for further discussion on this
  choice.}
For the special case of $m=0$ we set $S_0(y){:=} y$.

For any positive integer $m$ we define the function
$A_m:[1,\infty)\map\R_{>0}$ with
\begin{equation}
\label{eq:A_def}
A_m(x){:=}
\left[\frac{S_m(x)}{x^{m+1}}\right]^{-1}
=\left(\frac{1}{m+1}+\frac{1}{2x}\right)^{-1}
=\frac{2(m+1)x}{2x+m+1};
\end{equation}
for $m=0$ in particular, this gives $A_0(x){=} 1$.
Observe that $A_m$ is strictly increasing (in $x$) for all $m\geq 1$,
\begin{equation}
\label{eq:A_def_boundaries}
A_{m}(1)=\frac{2(m+1)}{m+3}\in [1,2),
\qquad\text{and}\qquad
\lim_{x\to\infty}A_m(x)=m+1.
\end{equation}
For the special case of $m=0$ we simply have $A_{0}(x)=1$ for all $x\geq 0$.
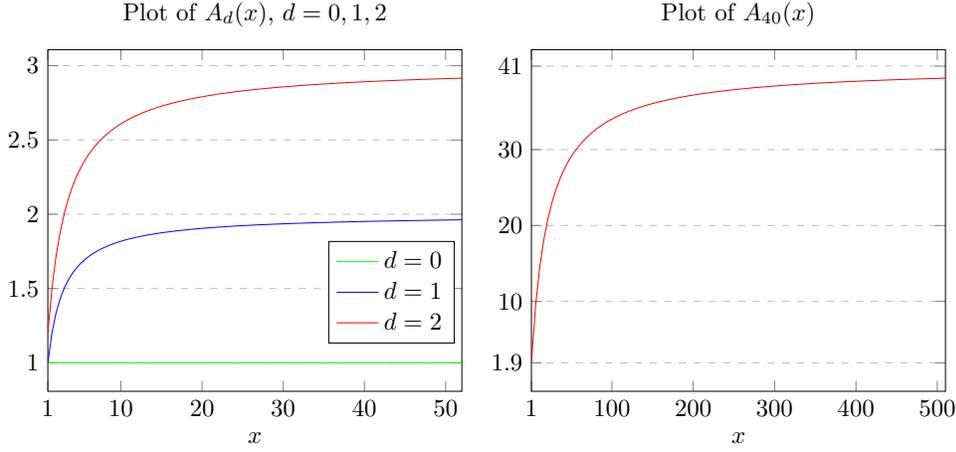
\begin{figure}
\centering
\scalebox{0.93}{
\begin{tikzpicture}
\begin{axis}[width=7.5cm, title={Plot of $A_d(x)$, $d=0,1,2$}, legend style={at={(0.83,0.44)},anchor=north}, xtick={1,10,20,30,40,50}, xmin=1, xmax=52, ytick={1,1.5,2,2.5,3}, ymajorgrids=true, grid style=dashed, ,xlabel={$x$}, legend entries={$d=0$,$d=1$,$d=2$}]
\addplot[color=green,domain=1:52,samples=10]{1};
\addplot[color=blue,domain=1:52,samples=100]{1/(1/(1+1)+1/(2*x))};
\addplot[color=red,domain=1:52,samples=100]{1/(1/(1+2)+1/(2*x))};
\end{axis}
\end{tikzpicture}}
\scalebox{0.93}{
\begin{tikzpicture}
\begin{axis}[width=7.5cm, title={Plot of $A_{40}(x)$}, legend pos=outer north east, xtick={1,100,200,300,400,500}, xmin=1, xmax=510, ytick={1.9,10,20,30,41}, ymajorgrids=true, grid style=dashed,xlabel={$x$}]
\addplot[color=red,domain=1:510,samples=100]{1/(1/(1+40)+1/(2*x))};
\end{axis}
\end{tikzpicture}}
\caption{Plots of functions $A_d$ for $d=0,1,2$ (left) and $d=40$ (right). For
$d\geq 1$ they are strictly increasing, starting at
$A_d(1)=\frac{2(d+1)}{d+3}\in[1,2)$ and going up to $d+1$ at the limit. Here,
$A_0(1)=1$, $A_1(1)=1$, $A_2(1)=6/5=1.2$ and $A_{40}(1)=82/43\approx 1.907$.}
\label{fig:pot_mono}
\end{figure}
\Cref{fig:pot_mono} shows a graph of these functions.
Since $A_m$ is strictly increasing for $m\geq 1$, its inverse function, $A^{-1}_m:[2\frac{m+1}{m+3},m+1)\map [1,\infty)$, is well-defined and also strictly increasing for all $m\geq 1$, with
\begin{equation}
\label{eq:A_inverse}
A^{-1}_m(x)=\frac{(m+1)x}{2(m+1-x)}.
\end{equation}

The following two lemmas (whose proofs can be found in
\cref{app:potential_mono_seq_monotonicity_proof,app:potential_mono_bounds_proof})
describe some useful properties regarding the algebraic behaviour, and the relation
among, functions $A_m$ and $S_m$:

\begin{lemma}
\label{lemma:potential_mono_seq_monotonicity}
Fix any reals $y\geq x\geq 1$. Then the sequences
$\frac{A_m(x)}{m+1}$ and $\frac{A_m(x)}{A_m(y)}$
are decreasing, and sequence $A_m(x)$ is increasing (with respect to $m$).
\end{lemma}

\begin{lemma}
\label{lemma:potential_mono_bounds}
Fix any integer $m\geq 0$ and reals $\gamma, w \geq 1$. Then
\begin{equation}
\label{eq:potential_mono_bounds_1}
\frac{\gamma^{m+1}}{A_m(\gamma w)}\leq \frac{S_m(\gamma(x+w))-S_m(\gamma x)}{w(x+w)^m}\leq  \gamma^{m+1},
\qquad\text{for all} \;\; x\geq 0,
\end{equation}
and
\begin{equation}
\label{eq:potential_mono_bounds_2}
\frac{\gamma^{m+1}}{m+1}\leq\frac{S_m(\gamma x)}{x^{m+1}} \leq \frac{\gamma^{m+1}}{A_m(\gamma)},
\qquad\text{for all}\;\; x\geq 1.
\end{equation}
\end{lemma}

\subsection{The Upper Bound}
\label{sec:upper_bound_therem}
Now we are ready to state our main positive result:
\begin{theorem}
\label{th:PoS_upper_general}
At any congestion game with polynomial latency functions of degree at most $d\geq 1$
and player weights ranging in $[1,W]$, for any
$\frac{2(d+1)W}{2W+d+1}\leq \alpha \leq d+1$
there exists an $\alpha$--approximate pure Nash equilibrium that, furthermore, has Price of Stability at
most
$$
1+\left(\frac{d+1}{\alpha}-1\right)W.
$$
\end{theorem}
Observe that, as the approximation parameter $\alpha$ increases, the Price of
Stability decreases, in a smooth way, from $\frac{d+3}{2}$ down to the optimal value
of $1$. Furthermore, notice how the interval within which $\alpha$ ranges, shrinks
as the range of player weights $W$ grows; in particular, its left boundary
$\frac{2(d+1)W}{2W+d+1}$ goes from $2\frac{d+1}{d+3}=2-\frac{4}{d+3}$ (for $W=1$) up
to $d+1$ (for $W\to\infty$).

As a result, \cref{th:PoS_upper_general} has two interesting
corollaries, one for $\alpha=\frac{2(d+1)W}{2W+d+1}$ and one for $W=1$ (unweighted games):
\begin{corollary}
\label{th:PoS_upper_general_exact_only}
At any congestion game with polynomial latencies of degree at most
$d\geq 1$ where player weights lie within the range $[1,W]$, there is
an $\frac{2(d+1)W}{2W+d+1}$-approximate pure Nash equilibrium with Price of Stability
at most $\frac{d+3}{2}$.
\end{corollary}
It is interesting to point out here that, in light of \cref{th:PoS_lower_approx}, the above result
of \cref{th:PoS_upper_general_exact_only} is almost asymptotically tight as
far as the Price of Stability is concerned (see the discussion preceding \cref{th:PoS_lower_approx}).
\begin{corollary}
\label{th:PoS_upper_general_unweighted}
At any \emph{unweighted} congestion game with polynomial latencies of degree at most $d\geq
1$, the Price of Stability of $\alpha$-approximate equilibria is at most
$\frac{d+1}{\alpha}$, for any $2\frac{d+1}{d+3}\leq \alpha \leq d+1$.
\end{corollary}

Before proving~\cref{th:PoS_upper_general}, we first restate it in the following
equivalent form, that parametrizes the approximation factor of the equilibrium, as
well as its Price of Stability guarantee, with respect to an ``external'', seemingly
artificial parameter $\gamma\in[1,\infty)$. The equivalence of the two formulations
is formally proven in \cref{app:restatement_main_upper}.
\begin{claim}[Restatement of~\cref{th:PoS_upper_general}]
\label{claim:PoS_upper_general}
For any $\gamma \geq 1$ there exists an $A_d(\gamma W)$-approximate pure Nash
equilibrium, which furthermore has Price of Stability at most
$\frac{d+1}{A_d(\gamma)}$, where $A_d$ is the strictly increasing
function\footnote{See \cref{fig:pot_mono}.} taking values within
$[2\frac{d+1}{d+3},d+1)$ defined in \eqref{eq:A_def}.
\end{claim}
The statement of~\cref{claim:PoS_upper_general} may at first seem a bit cryptic,
compared to~\cref{th:PoS_upper_general}. Nevertheless, it brings forth some
important aspects of our upper bound construction that are not immediately obvious
from~\cref{th:PoS_upper_general}. In particular, notice how the weight range $W$ has
\emph{no} effect in the Price of Stability guarantee in the statement
of~\cref{claim:PoS_upper_general}, but appears only in the approximation factor of
the equilibrium. Furthermore, as it will become more clear in
\cref{sec:full_faulhaber_note}, this formulation provides a good degree of high-level
abstraction that helps with generalizing and improving our result in certain cases,
in a unified way. We believe this is important, since it is a promising direction
for future work (see also the discussion in~\cref{sec:EulerMaclaurin}). 

\begin{proof}[Proof of~\cref{claim:PoS_upper_general}]
Without loss of generality, it is enough to consider only weighted congestion games with
\emph{monomial} latency functions (of degree at most $d$); any polynomial is a sum
of monomials, so we can just simulate the polynomial latency of a facility by
introducing monomial-latency facilities for each one of its summands. More formally,
if a facility $e$ has latency function $ c_e(x)=\sum_{j=0}^d a_{e,j}x^j $, with
constants $a_{e,0},a_{e,1},\dots,a_{e,d}\geq 0$, we can replace $e$ by facilities
$e_0,\dots,e_d$ with latencies $c_{e_j}(x)=a_{e,j}x^j$, without any change to the
costs of the players. Furthermore, we can safely ignore all such facilities $e_j$ with $a_{e,j}=0$, since they have absolutely no effect in the players' costs.

So, from now on assume that for each facility $e\in E$ there exists a real constant
$a_{e}> 0$ and an nonnegative integer $m_e\leq d$ such that
$$
c_e(x)=a_{e}x^{m_e}.
$$
Then, in order to utilize \cref{lemma:approximate_PoS_ratios}, we choose functions
\begin{equation}
\label{eq:pot_def_poly}
\phi_e(x)
=a_{e} \cdot S_{m_e}(\gamma x),
\end{equation}
were $\gamma$ is a real parameter, free to range in $[1,\infty)$.
Recall here that functions $S_m$ and $A_m$ are defined in
\eqref{eq:potential_mono_def} and \eqref{eq:A_def}. 
To simplify notation, from now on we fix an arbitrary facility $e$ and drop the $e$-subscripts from $\phi_e$, $c_e$, $a_e$ and $m_e$.

From
\eqref{eq:potential_mono_bounds_1} of \cref{lemma:potential_mono_bounds} we get that, for any $x\geq
0$ and $w\in[1,W]$,
\begin{align*}
\frac{\gamma^{m+1}}{A_m(\gamma w)}\leq \frac{\phi(x+w)-\phi(x)}{w\cdot c(x+w)}
&=\frac{a\left[S_m(\gamma(x+w))-S_m(\gamma x)\right]}{w\cdot a\cdot  (x+w)^m}
\leq \gamma^{m+1}.
\end{align*}
Similarly, from \eqref{eq:potential_mono_bounds_2} we have that for any $x\geq 1$,
\begin{align*}
\frac{\gamma^{m+1}}{m+1}\leq \frac{\phi(x)}{x\cdot c(x)}
= \frac{a\cdot S_m(\gamma x)}{x\cdot a\cdot x^m}=\frac{S_m(\gamma x)}{x^{m+1}}\leq \frac{\gamma^{m+1}}{A_m(\gamma)}.
\end{align*}

Now let us just scale the functions $\phi_e$ we defined in \eqref{eq:pot_def_poly}
by a factor of $\frac{1}{S_{m}(\gamma)}$ and define a potential function
\begin{equation*}
\bar \phi(x)=\frac{\phi(x)}{S_m(\gamma)}=\frac{A_m(\gamma)}{\gamma^{m+1}}\cdot \phi(x)=a\frac{A_m(\gamma)}{\gamma^{m+1}} S_m(\gamma x).
\end{equation*}
Our previous bounds for $\phi$ show us that $\bar\phi$ satisfies
the requirements of \cref{lemma:approximate_PoS_ratios} with parameters
\begin{align*}
\alpha_1 &=\frac{\gamma^{m+1}}{A_m(\gamma w)} \cdot \frac{A_m(\gamma)}{\gamma^{m+1}}=\frac{A_m(\gamma)}{A_m(\gamma w)}\geq \frac{A_d(\gamma)}{A_d(\gamma w)}\geq \frac{A_d(\gamma)}{A_d(\gamma W)}\\
\alpha_2 &= \gamma^{m+1}\cdot \frac{A_m(\gamma)}{\gamma^{m+1}}=A_m(\gamma)\leq A_d(\gamma)\\
\beta_1 &= \frac{\gamma^{m+1}}{m+1} \cdot \frac{A_m(\gamma)}{\gamma^{m+1}}=\frac{A_m(\gamma)}{m+1}\geq \frac{A_d(\gamma)}{d+1}\\
\beta_2 &=\frac{\gamma^{m+1}}{A_m(\gamma)} \cdot \frac{A_m(\gamma)}{\gamma^{m+1}}=1,
\end{align*}
where the inequalities hold due to \cref{lemma:potential_mono_seq_monotonicity},
taking into consideration the fact that $\gamma w\geq \gamma\geq 1$ and $m\leq d$;
specifically for the last inequality on the bound of $\alpha_1$ we also used the
fact that $A_d$ is monotonically increasing.

Putting everything together, from \cref{lemma:approximate_PoS_ratios} we deduce that
indeed there exists an $A_d(\gamma W)$--approximate pure Nash equilibrium with Price
of Stability at most $\frac{d+1}{A_d(\gamma)}$. The fact that $A_d(\gamma)$ ranges
(monotonically) in $[2\frac{d+1}{d+3},d+1)$ is a consequence
of~\eqref{eq:A_def_boundaries}.
\end{proof}

\subsection{Small vs Large Degree Polynomials}
\label{sec:full_faulhaber_note}

One can argue that our choice to keep only the first two terms in Faulhaber's
formula \eqref{eq:faulhaber_formula_expand}, when defining our approximate potential
in \eqref{eq:potential_mono_def}, is suboptimal. To some extent, this is correct; it
is exactly the reason why this seemingly ``unnatural'' lower bound of
$2\frac{d+1}{d+3}=2-\frac{4}{d+3}$ for the approximation parameter $\alpha$ appears
in \cref{th:PoS_upper_general_unweighted} (or, more generally,
$\frac{2(d+1)W}{2W+d+1}$ in \cref{th:PoS_upper_general}). It would be nicer if
$\alpha$ could simply start from $1$ instead. Indeed, this can be achieved for small
values of $d$, as described below.

Considering the entire right-hand side expression in
\eqref{eq:faulhaber_formula_expand}, one can take the full, exact version of
Faulhaber's formula, that can be written\footnote{See, e.g.,
\citep[p.~288]{Knuth1993} or \citep[Eq.~23.1.4]{Abramowitz1964}.} in a very elegant
way as
\begin{equation}
\sum_{k=1}^n k^m
	= \frac{1}{m+1}\left[B_{m+1}(n+1)-B_{m+1}(0) \right],\label{eq:faulhaber_formula_full}
\end{equation}
where
$$
B_m(y)=\sum_{k=0}^m\binom{m}{k}B_ky^{m-k}, \qquad y\geq 0,
$$
are the Bernoulli polynomials, and coefficients $B_k=B_k(0)$ are the standard
Bernoulli numbers we used before. Now we can use \eqref{eq:faulhaber_formula_full}
to define a more fine-tuned version for $S_m$, that is, for $m\geq 1$ set $\hat
S_m(x)=\frac{1}{m+1}\left[B_{m+1}(x+1)-B_{m+1} \right]$ instead of
\eqref{eq:potential_mono_def}. For example, for degrees up to $m\leq 4$ these new
polynomials are:
\begin{gather*}
\hat S_0(x)=x,\qquad \hat S_1(x)= \frac{1}{2} x (x+1),\qquad \hat S_2(x)= \frac{1}{6} x (2 x^2+3 x+1)\\
\hat S_3(x)= \frac{1}{4} x^2 (x+1)^2,
\qquad
\hat S_4(x)= \frac{1}{30} x (6 x^4+15 x^3+10 x^2-1)
\end{gather*}
Using these values, one can verify that for up to $m\leq 4$, all our critical
technical requirements for the proof of \cref{claim:PoS_upper_general} (and thus,
\cref{th:PoS_upper_general} itself) are satisfied: most notably
\cref{lemma:potential_mono_seq_monotonicity,lemma:potential_mono_bounds}, and the
monotonicity of $\hat A_m(x)=\frac{x^{m+1}}{\hat S_m(x)}$ (with respect to $x\geq
1$). In particular, now we have that $\hat A_m(1)=\frac{1^{m+1}}{\hat S_m(1)}=1$,
which is exactly what we wanted: it means that the critical quantities $\hat
A_d(\gamma W)$ and $\hat A_d(\gamma)$ in \cref{claim:PoS_upper_general} can start
taking values all the way down to $\hat A_d(W)$ and $\hat A_d(1)=1$, respectively.
This translates to the approximation ratio parameter $\alpha$ in our main result in
\cref{th:PoS_upper_general} starting to range from $\alpha\geq
\hat A_d(W)$.

Thus,
\begin{quotation}
\emph{\cref{th:PoS_upper_general} can be rewritten for $d\leq 4$, with the
approximation parameter $\alpha$ taking values in $\hat A_d(W)\leq
\alpha\leq d+1$. In particular, for unweighted games, this means that
\cref{th:PoS_upper_general_unweighted} can be rewritten with $\alpha$ taking values
within the entire range of $[1,d+1]$.}
\end{quotation}

However, there is a catch, that does not allow us to do that in general; as $m$
grows large, the Bernoulli polynomials, that now play a critical role in our
definition of functions $\hat S_m$ (see \eqref{eq:faulhaber_formula_full}), start to
behave in a rather erratic, non-smooth way within the interior of the real intervals
between consecutive integer values. For example, one can check that, for $d=14$
function $\hat A_{14}$ is \emph{not} monotonically increasing within $[1,2]$. Even
more disastrously, for $d=20,21$ functions $\hat S_d$ take \emph{negative} values in
$[1,2]$ !

\appendix

\section{Lower Bound Proofs}

\subsection{Technical Lemmas}
\begin{lemma}
\label{lemma:newtech1} For any $d\geq 9$,
\begin{equation*}
\left(1+\frac{\ln d}{d}\right)^d \geq \frac{d}{\ln d}.
\end{equation*}
\end{lemma}
\begin{proof}
From~\citet[Eq.~(3), p.~267]{Mitrinovic1970} we know that the following inequality holds for all $n \geq 1$ and $1\leq x \leq n$:
$$
\left(1+\frac{x}{n}\right)^n \geq e^x\left(1-\frac{x^2}{n}\right).
$$
Applying it with $n=d$ and $x=\ln d$ we get that indeed
$$
\left(1+\frac{\ln d}{d}\right)^d 
\geq
e^{\ln d}\left(1-\frac{\ln^2 d}{d}\right)
\geq d \frac{1}{\ln d},
$$
the last inequality holding due to the fact that for $d\geq 9$, 
$
\frac{\ln^2 d}{d}+\frac{1}{\ln d} \leq \frac{\ln^2 (9)}{9}+\frac{1}{\ln (9)}\approx 0.992 \leq 1.
$
\end{proof}
\begin{lemma}
\label{lemma:ratio_monotone}
For any integer $d\geq 2$, the function
$f:(0,\infty)^2\map (0,\infty)$ defined by
$$
f(x,y)=\frac{(y+x+1)^d-(y+x)^d}{(y+1)^d-y^d}
$$
is monotonically decreasing with respect to $y$. Furthermore, for $d\geq 9$,
\begin{equation}
\label{eq:c_def_lemma_3}
\zeta^{d+1}\leq
f\left((\beta_d\Phi_d+1)(\zeta-1),\beta_d\Phi_d-(1-\beta_d)\zeta\right)
\qquad\text{for all}\;\;
 \zeta\in[1,2],
\end{equation}
where $\beta_d$ is defined in \cref{lemma:def_c}.
\end{lemma}
\begin{proof} First let us define function $h:(0,\infty)\map (0,\infty)$ with
\begin{equation}
\label{eq:def_h} h(t)=\frac{(t+1)^d-t^d}{(t+1)^{d-1}-t^{d-1}}.
\end{equation} We will show that $h$ is increasing, which will suffice to prove the
desired monotonicity of $f$ since its derivative is
\begin{align*}
\frac{\partial f(x,y)}{\partial y}
&=
\frac{d \left[(x+y+1)^{d-1}-(x+y)^{d-1}\right]}{(y+1)^d-y^d}\\
&\qquad\qquad\qquad -
\frac{d
\left[(y+1)^{d-1}-y^{d-1}\right]
\left[(x+y+1)^d-(x+y)^d\right]}{\left[(y+1)^d-y^d\right]^2}
\\
&=\frac{d\left[(y+1)^{d-1}-y^{d-1}\right]\left[(x+y+1)^{d-1}-(x+y)^{d-1}\right]}{\left[(y+1)^d-y^d\right]^2}\left[h(y)-h(x+y)
\right],
\end{align*} which is negative due to the monotonicity of $h$. To prove that $h$ is
indeed increasing, we will show something stronger; namely that function $\bar h:(1,\infty)\map
(0,\infty)$ with
\begin{equation}
\label{eq:def_h_bar}
\bar h(t)=\frac{t^d-(t-1)^d}{t^{d}-t(t-1)^{d-1}}
\end{equation} is increasing. This will suffice to demonstrate that $h$ is
increasing as well, since $h(t)=(t+1)\cdot\bar h(t+1)$.
Taking its derivative we see that
$$
\frac{\partial \bar h(t)}{\partial t}=
\frac{\left[(t-1)^d-t^{d}+d t^{d-1}\right](t-1)^d }{\left[t^{d+1}-t^d-t
(t-1)^d\right]^2}>0
$$
since from the convexity of function $t\mapsto t^d$ we know that $t^d-(t-1)^d<
dt^{d-1}$.

Now let us prove the
remaining part of our lemma, that is \eqref{eq:c_def_lemma_3}. Observe that if we
set $\zeta = 1$ to \eqref{eq:c_def_lemma_3} it is satisfied, since $f(0,y)=1$ for
any $y>0$. So, it is enough if we prove that
\begin{multline*}
\zeta^{-(d+1)}f\left((\beta_d\Phi_d+1)(\zeta-1),\beta_d\Phi_d-(1-\beta_d)\zeta\right)
=\\
\zeta^{-(d+1)}\frac{\left[(\alpha+\beta)\zeta\right]^d-\left[(\alpha+\beta)\zeta-1\right]^d}{\left[\alpha+1-(1-\beta)\zeta\right]^d-\left[\alpha-(1-\beta)\zeta\right]^d}
\end{multline*}
is increasing with respect to $\zeta\in [1,2]$, where here we are using $\beta=\beta_d$ and 
$\alpha=\beta \Phi_d$. So, if we define
\begin{align*} f_1(\zeta ) &=
\left[(\alpha+\beta)\zeta\right]^d-\left[(\alpha+\beta)\zeta-1\right]^d\\ f_2(\zeta)
&= \left[\alpha+1-(1-\beta)\zeta\right]^d-\left[\alpha-(1-\beta)\zeta\right]^d
\end{align*} and we compute the derivative $
\frac{\partial}{\partial
\zeta}\left(\zeta^{-(d+1)}\frac{f_1(\zeta)}{f_2(\zeta)}\right) $ of the above
expression, we need to show that
\begin{equation}
\label{eq:lemma_fraction_helper_3}
\zeta\left[\frac{f_1'(\zeta)}{f_1(\zeta)}-\frac{f_2'(\zeta)}{f_2(\zeta)}\right] \geq
d+1.
\end{equation} Now notice that
$$
\zeta\frac{f_1'(\zeta)}{f_1(\zeta)}=d(\alpha+\beta)\zeta\frac{\left[(\alpha+\beta)\zeta\right]^{d-1}-\left[(\alpha+\beta)\zeta-1\right]^{d-1}}{\left[(\alpha+\beta)\zeta\right]^d-\left[(\alpha+\beta)\zeta-1\right]^d}=\frac{d}{\bar
h\left((\alpha+\beta)\zeta\right)},
$$
where $\bar h$ is the increasing function defined in \eqref{eq:def_h_bar}, so taking into consideration that
$$
(\alpha+\beta)\zeta=\beta(\Phi_d+1)\zeta\leq\frac{1}{2}(\Phi_d+1)2\leq \Phi_d+1,
$$
we can get that
\begin{align*}
\zeta\frac{f_1'(\zeta)}{f_1(\zeta)} &\geq \frac{d}{\bar
h(\Phi_d+1)} =d(\Phi_d+1)\frac{(\Phi_d+1)^{d-1}-\Phi_d^{d-1}}{(\Phi_d+1)^{d}-\Phi_d^{d}}\\
&= d\frac{\Phi_d^{d+1}-\Phi_d^{d}-\Phi_d^{d-1}}{\Phi_d^{d+1}-\Phi_d^{d}}
= d-\frac{d}{\Phi_d^2-\Phi_d}.
\end{align*}
Similarly, we can see that
$$
-\zeta\frac{f_2'(\zeta)}{f_2(\zeta)}
=
\frac{d(1-\beta)\zeta}{h(\alpha-(1-\beta)\zeta)},
$$
where $h$ is the increasing function defined in \eqref{eq:def_h}, so taking into consideration that
$$
\alpha-(1-\beta)\zeta\leq\beta\Phi_d-(1-\beta)\leq\frac{\Phi_d-1}{2}
\qquad\text{and}\qquad (1-\beta)\zeta\geq \frac{1}{2},
$$
we get that
$$
-\zeta\frac{f_2'(\zeta)}{f_2(\zeta)}
\geq
\frac{d/2}{h\left((\Phi_d-1)/2\right)}
=
d\frac{(\Phi_d+1)^{d-1}-(\Phi_d-1)^{d-1}}{(\Phi_d+1)^{d}-(\Phi_d-1)^{d}}.
$$
Putting everything together, in order to prove the desired
\eqref{eq:lemma_fraction_helper_3}, it now suffices to show that
$$
d\frac{(\Phi_d+1)^{d-1}-(\Phi_d-1)^{d-1}}{(\Phi_d+1)^{d}-(\Phi_d-1)^{d}}-\frac{d}{\Phi_d^2-\Phi_d}\geq
1,
$$
which we know holds from \eqref{eq:asymptotic_phi_3} of
\cref{lemma:lower_bound_phi_asymptotic}.
\end{proof}

\begin{lemma}
\label{lemma:lower_bound_phi_asymptotic}
For any integer $d\geq 9$,
\begin{equation}
\label{eq:asymptotic_phi_3}
\frac{(\Phi_d+1)^{d-1}-(\Phi_d-1)^{d-1}}{(\Phi_d+1)^{d}-(\Phi_d-1)^{d}}-\frac{1}{\Phi_d^2-\Phi_d}\geq
\frac{1}{d}.
\end{equation}
Furthermore, asymptotically $\Phi_d\sim\frac{d}{\ln d}$, i.e.
\begin{equation}
\label{eq:phiasy}
\lim_{d\to\infty}\frac{\Phi_d}{d/\ln d}=1.
\end{equation}
In particular, for any integer $d$,
\begin{equation}
\label{eq:asymptotic_phi_2}
 \Phi_d \leq \gamma_d\frac{d}{\ln d}
 \qquad\text{with}\quad
 \gamma_d{:=} \frac{\ln d}{\mathcal W(d)}\leq 1.368
 \quad\text{and}\quad
 \lim_{d\to\infty} \gamma_d=1,
\end{equation} where $\mathcal W(\cdot)$ denotes the (principal branch of the) Lambert--W
function\footnote{That is, for any positive real $x$, $\mathcal W(x)=z$ gives the unique
positive real solution $z$ to the equation $x=z\cdot e^z$.}.
\end{lemma}
\begin{proof}
To show \eqref{eq:asymptotic_phi_3}, first one can simply
numerically verify that it indeed holds for all integers $d=9,10,\dots,17$, so let us
just focus on the case when $d\geq 18$. For simplicity, in the remainder of the
proof we denote $y=\Phi_d$. It is easy to see\footnote{Here we are silently using
the fact that $\Phi_n$ is an increasing function of the integer $n$. One can
formally prove this by, e.g.,\ combining Lemmas 5.1 and 5.2 of \citet{Aland2011}.}
then that
$$
y\geq \Phi_{18}\approx 8.11>8.
$$
Performing some elementary algebraic manipulations in \eqref{eq:asymptotic_phi_3},
we can equivalently write it as
$$
(y+1)^{d-1}\left[-y^3+dy^2-(2d-1)y-d\right]\geq (y-1)^{d+1}(d-y).
$$
Using the fact that $(y+1)^{d-1}=\frac{y^{d+1}}{y+1}$, and then that $y^{d+1}\geq (y-1)^{d+1}$, we
can see that it is enough to show that
$$
-y^3+dy^2-(2d-1)y-d \geq (d-y)(y+1),
$$
or equivalently,
\begin{gather}
y^2(d-y) -2dy + y -d \geq (d-y)(y+1) \notag\\
(d-y)\left[y^2-1-(y+1)\right] \geq 2dy \notag\\
(d-y)(y+1)(y-2) \geq 2y d \notag\\
\frac{(y+1)(y-2)}{2y} \geq \frac{1}{1-\frac{y}{d}}. \label{eq:helpernew1}
\end{gather}
For the last inequality we took into consideration that $\frac{y}{d}<1$. As a matter of fact, using an upper
bound of $y\leq \frac{2d}{\ln d}$ on $y=\Phi_d$ (see \citep[Lemma 5.4]{Aland2011} or~\eqref{eq:asymptotic_phi_2}), we have that
$$
\frac{y}{d}\leq \frac{2}{\ln d} \leq \frac{2}{\ln(18)}\approx 0.692.
$$
Due to the fact that function $z\mapsto\frac{1}{1-z}$ is increasing for $z\in
[0,1)$, this bound gives us that $\frac{1}{1-\frac{y}{d}} \leq \frac{1}{1-0.69} \leq
3.226$. In a similar way, noticing that function $z\mapsto\frac{(z+1)(z-2)}{2z}$ is
increasing for $z>2$, using the fact that $y>8$ we can derive that
$\frac{(y+1)(y-2)}{2y}>\frac{(8+1)(8-2)}{2\cdot 8}=\frac{27}{8}\approx 3.375$. This
establishes the validity of \eqref{eq:helpernew1}.

Next, we deal with upper-bounding the values of $\Phi_d$ and proving
\eqref{eq:asymptotic_phi_2}, since we'll need this for establishing the asymptotics
of $\Phi_d$ in~\eqref{eq:phiasy}. Here we will make use of the following property,
which was shown in~\citet[Lemma~5.2,Theorem~3.4]{Aland2011}: for any real
$\gamma>0$,
\begin{equation}
\label{eq:boundsphiequiv}
\Phi_d<\gamma \frac{d}{\ln d}
\quad
\ifif
\quad
\left(1+\frac{\ln d}{\gamma d}\right)^d < \frac{\gamma d}{\ln d}.
\end{equation}
In a similar way to the proof of~\citet[Theorem~3.4]{Aland2011}, using a binomial expansion we can compute:
\begin{align*}
\left(1+\frac{\ln d}{\gamma d}\right)^d 
&= \sum_{k=0}^d\binom{d}{k}\left(\frac{\ln d}{\gamma d}\right)^k
= \sum_{k=0}^d\frac{d!}{d^k(d-k)!}\left(\frac{\ln d}{\gamma}\right)^k\frac{1}{k!}\\
&= \sum_{k=0}^d\left(1-\frac{1}{d}\right)\left(1-\frac{2}{d}\right)\dots\left(1-\frac{k-1}{d}\right)\left(\frac{\ln d}{\gamma}\right)^k\frac{1}{k!}\\
&< \sum_{k=0}^\infty \frac{\left(\frac{\ln d}{\gamma}\right)^k}{k!}
=e^{\frac{\ln d}{\gamma}}=d^{1/\gamma},
\end{align*}
where for the second to last equality we used the power series representation of the
exponential function: $e^z=\sum_{k=0}^\infty\frac{z^k}{k!}$. Thus, from
\eqref{eq:boundsphiequiv} we can derive that, for any $\gamma>0$,
\begin{equation}
 \label{eq:lemma_aland} d^{1/\gamma}\leq \frac{\gamma\cdot d}{\ln d} \quad\then\quad
\Phi_d<\gamma \frac{d}{\ln d},
\end{equation} 
Using $\gamma_d=\frac{\ln d}{\mathcal W(d)}$
as defined in the statement of our lemma, we compute:
$$
d^{1/\gamma_d}=d^{\frac{\mathcal W(d)}{\ln d}}=e^{\mathcal W(d)}
$$
and
$$
\gamma_d\frac{d}{\ln d}=\frac{\ln d}{\mathcal W(d)}\frac{d}{\ln d}=\frac{d}{\mathcal W(d)}.
$$
Thus, since from the definition of function $\mathcal W$ we know that
\begin{equation}
\label{eq:lambert}
\mathcal W(d)e^{\mathcal W(d)}=d,
\end{equation}
we deduce that $\gamma=\gamma_d$ indeed satisfies the left hand
side of
\eqref{eq:lemma_aland}, giving us the desired upper bound for $\Phi_d$.

For the asymptotic behaviour of $\gamma_d$ when $d$ grows large, observe that by
taking logarithms in \eqref{eq:lambert} we get $$\mathcal W(d) + \ln \mathcal W(d) =
\ln d$$ and so
$$
\lim_{d\to\infty} \gamma_d =\lim_{d\to\infty}\frac{\ln d}{\mathcal W(d)}
=\lim_{d\to\infty}\left[\frac{\ln \mathcal W(d)}{\mathcal W(d)}+1
\right]=1+\lim_{z\to\infty}\frac{\ln z}{z}=1,
$$
since it is easy to see that $\lim_{d\to\infty} \mathcal W(d)=\infty$.

Finally, let us now establish~\eqref{eq:phiasy}. Due to~\eqref{eq:asymptotic_phi_2}
that we have already proved, it is enough to just show a lower bound of
$\lim_{d\to\infty}\frac{\Phi_d}{d/\ln d} \geq 1$. We will do this by showing that
$\Phi_d\geq \frac{d}{\ln d}$ for sufficiently large values of $d$. Indeed, by
\eqref{eq:boundsphiequiv} this is equivalent to proving that
\begin{equation*}
\left(1+\frac{\ln d}{d}\right)^d \geq \frac{d}{\ln d},
\end{equation*}
which from \cref{lemma:newtech1} we know it holds for all $d\geq 9$.
\end{proof}

\subsection{Proof of \texorpdfstring{\cref{lemma:def_c}}{Lemma~3.2}}
\label{app:def_c_proof}

To decongest notation a bit in the proof, we will drop the $d$
subscripts from $c_d$ and $\beta_d$ whenever this is causing no confusion. Starting
with
\eqref{eq:c_def_lemma_1}, if we solve with respect to $\beta$ we get
\begin{equation}
\label{eq:c_def_lemma_helper2}
\Phi_d+\frac{1}{\beta}\geq \Phi_d^{1+\frac{2}{d}}
\quad
\ifif 
\quad
\beta \leq \frac{1}{\Phi_d(\Phi_d^{2/d}-1)}
=\frac{\Phi_d}{2\Phi_d+1},
\end{equation}
where for the last equality we used the fact that
$$
\Phi_d^{2/d}=\left(1+\frac{1}{\Phi_d}\right)^2=\frac{1}{\Phi_d^2}+\frac{2}{\Phi_d}+1
$$
which is a direct consequence of the definition of $\Phi_d$:
$$
\Phi_d^{d+1}=(\Phi_d+1)^{d}
\quad\ifif \quad
\Phi_d^{1+\frac{1}{d}}=\Phi_d+1
\quad\ifif\quad
\Phi_d^{1/d}=1+\frac{1}{\Phi_d}.
$$
Substituting \eqref{eq:b_def} into \eqref{eq:c_def_lemma_helper2},
\begin{align*}
1-\Phi_d^{-c}\leq \frac{\Phi_d}{2\Phi_d+1}
& \ifif
\Phi_d^{-c}\geq \frac{\Phi_d+1}{2\Phi_d+1}\\
& \ifif c\leq \frac{\ln(2\Phi_d+1)-\ln(\Phi_d+1)}{\ln \Phi_d},
\end{align*}
which holds by the very definition of $c$ in \eqref{eq:c_def} if we relax the floor
operator.

Let us now move to \eqref{eq:c_def_lemma_2}, and in particular  lower-bound the
values of parameter $\beta$, as $d$ grows large. Due to the floor operator in
\eqref{eq:c_def}, parameter $c$ can be lower bounded by
\begin{equation*} c \geq
\frac{\ln(2\Phi_d+1)-\ln(\Phi_d+1)}{\ln \Phi_d}
-\frac{1}{d}=\log_{\Phi_d}\frac{2\Phi_d+1}{\Phi_d+1}
-\frac{1}{d}
\end{equation*} and since $\beta$ is increasing with respect to $c$,
\begin{equation}
\label{eq:def_c_helper_1}
\beta = 1-\Phi_d^{-c}
\geq 1- \frac{\Phi_d+1}{2\Phi_d+1}\Phi_d^{1/d}= 1- \frac{\Phi_d+1}{2\Phi_d+1}\left(1+\frac{1}{\Phi_d}\right).
\end{equation}
Taking limits and recalling that $\lim_{d\to\infty} \Phi_d=\infty$, we get the
desired lower bound of
$$
\lim_{d\to\infty} \beta\geq 1-\frac{1}{2}\cdot (1+0)=\frac{1}{2}.
$$

The upper bound of $\beta\leq \frac{1}{2}$ can be easily derived by \eqref{eq:c_def_lemma_helper2}:
$$
\beta \leq \frac{\Phi_d}{2\Phi_d+1} \leq \frac{\Phi_d}{2\Phi_d}=\frac{1}{2}.
$$

For small values of $d$, and in particular in order to prove that $\beta\geq0.38$,
one can numerically compute the values for $\beta$ directly from \eqref{eq:c_def}.
For example, for $d=9,\dots,100$, these values are shown in \cref{fig:plot_b}. The
lower and upper red lines in \cref{fig:plot_b} correspond to the relaxation of the
floor operator we used in the lower and upper bounds for $\beta$ in
\eqref{eq:def_c_helper_1} and
\eqref{eq:c_def_lemma_helper2}, respectively. The actual values of $\beta$ lie
between these two lines. Using these values and the resulting monotonicity for
$\beta$, one can also prove the lower bound of $3$ for $dc$, by observing that by
setting $d=9$ in \eqref{eq:c_def} we have that for any $d\geq 9$
$$
d\cdot c \geq \left\lfloor
9\frac{\ln(2\cdot\Phi_9+1)-\ln(\Phi_9+1)}{\ln(\Phi_9)}
\right\rfloor\approx\lfloor 3.368\rfloor=3,
$$
since $\Phi_9\approx 5.064$.
\subsection{Proof of \texorpdfstring{\cref{claim:best_s_i^*}}{Claim~3.3}}
\label{app:PoS_lower_6_proof}

Since we have fixed that $s_j=\tilde s_j$ for all
players $j< i$ and also the strategies of players $j>i+\mu$ have no effect on the
cost of player $i$ (and in particular in \eqref{prop:dominating_prof_3}), it is safe
if we briefly abuse notation and for now assume that $\vecc
s_{-i}=(s_{i+1},\dots,s_{i+\mu})$.

We generate the desired dominating profile $\vecc s'_{-i}$ inductively, by running
Procedure~$\algoname{Dominate}(\vecc s_{-i},i)$ described formally below, scanning and
modifying profile $\vecc s_{-i}$ from right to left.
\begin{algorithm2e}
\NoCaptionOfAlgo
\KwIn{Profile $\vecc s_{-i}=(s_{i+1},\dots,s_{i+\mu})$; Player
$i\in\sset{\mu+1,\dots,n}$}
\KwOut{Profile $\vecc s'_{-i}=(s_{i+1}',\dots,s_{i+\mu}')$ of the form described in
\cref{prop:dominating_prof_1,prop:dominating_prof_2,prop:dominating_prof_3} of
\cref{claim:best_s_i^*}, that satisfies~\cref{eq:helper_PoS_lower_6}}
\ShowLn $\vecc s_{-i}' \gets \vecc s_{-i}$\;
\ShowLn $s_{i+\mu}' \gets s_{i+\mu}^*$\; \label{proc_line:dominate_last_player}
\ShowLn $k \gets i+\mu -1$\;
\ShowLn \While{exists $j\in\sset{i+1,\dots,k-1}$ such that $s'_j= s_j^*$
\label{proc_line:dominate_while_1}}{
\ShowLn $s'_k \gets \tilde s_k$\; \label{proc_line:dominate_while_2}
\ShowLn $k \gets k-1$\; \label{proc_line:dominate_while_3}}
\caption{\textbf{Procedure} $\algoname{Dominate}(\vecc s_{-i},i)$}
\label{proc:dominate}
\end{algorithm2e}

First, it is not difficult to see that the output profile $\vecc s'_{-i}$ of
$\algoname{Dominate}(\vecc s_{-i},i)$ indeed has the desired format described in
\cref{prop:dominating_prof_1,prop:dominating_prof_2,prop:dominating_prof_3} of~\cref{claim:best_s_i^*}. In particular, after any execution of the
while-loop in lines
\ref{proc_line:dominate_while_1}--\ref{proc_line:dominate_while_3} of
Procedure~\algoname{Dominate}, $s_j'=\tilde s_j$ for any $j=k+1,\dots,i+\mu-1$.
Furthermore, it is also easy to see that switching player's $i+\mu$ strategy to
$s'_{i+\mu}= s^*_{i+\mu}$ can only increase player's $i$ cost,
i.e.~\eqref{eq:helper_PoS_lower_6} is satisfied after line
\ref{proc_line:dominate_last_player} of \algoname{Dominate}: if player $i+\mu$
chooses $\tilde s_{i+\mu}$ instead, she contributes nothing to the cost of player
$i$, since she does not put her weight in any of the facilities $i+1,\dots, i+\mu$
played by player $i$.

So, it remains to be shown that after every iteration of the while-loop, condition
\eqref{eq:helper_PoS_lower_6} is maintained. Since in any such loop only the
strategy of player $k$ is possibly switched from $s_k^*$ to $\tilde s_k$, it is
enough if we show that
$
C_i(\tilde s_k,\vecc s'_{-k}) \geq C_i(s^*_k,\vecc s'_{-k})
$
or, since for any facility $j<k$ it holds that $x_j(\tilde s_k,\vecc
s'_{-k})=x_j(s^*_k,\vecc s'_{-k})$, equivalently
$$
\sum_{j=k}^{i+\mu}c_j(x_j(\tilde s_k,\vecc s'_{-k}))
\geq
\sum_{j=k}^{i+\mu}c_j(x_j(s^*_k,\vecc s'_{-k})).
$$
If we let $z_j$, for any $j\geq k$, denote the load on facility $j$ induced by every
player \emph{except} from player $k$, that is formally
$$
z_j = \sum \sset{w_\ell\fwh{\ell\in\ssets{j-\mu,\dots,j}\setminus\ssets{k}
\;\;\land\;\; j\in s_\ell'}},
$$
the above can be written as
$$
\sum_{j=k+1}^{i+\mu} \left[ c_{j}(z_j+w_k)-c_j(z_j) \right]
\geq c_k(z_k+w_k)-c_k(z_k).
$$
Thus, it is sufficient only take $j=i+\mu$ in the above sum and just prove that
$$c_{i+\mu}(z_{i+\mu}+w_k)-c_{i+\mu}(z_{i+\mu}) \geq c_k(z_k+w_k)-c_k(z_k),$$ which
is equivalent to
\begin{equation}
\label{eq:helper_PoS_lower_3}
w^{-(i+\mu-k)(d+1)}\frac{(z_{i+\mu}+w_k)^d-z_{i+\mu}^d}{(z_k+w_k)^d-z_k^d}\geq 1.
\end{equation}

If we define
$$
A=\sset{j\in\ssets{i+1,\dots,k-1}\fwh{s_j=s_j^*}},
$$
that is, $A$ is the set of players below $k$ (and above $i$) that do \emph{not}
contribute with their weight to the cost of players $k$ and $i+\mu$, we have that
$$
z_k=\sum_{\underset{j\notin A}{j=k-\mu}}^{k-1}w_j=\alpha w_k -\sum_{j\in A} w_j
$$
and
$$
z_{i+\mu}=\sum_{\underset{j\notin A\union\ssets{k}}{j=i}}^{i+\mu}w_j=(\alpha+1)
w_{i+\mu} -w_k -\sum_{j\in A} w_j,
$$
because by our inductive process we know that $s_j'=\tilde s_j$ for all
$j=k+1,\dots,i+\mu-1$. Thus
$$
z_{i+\mu}=z_k+(\alpha+1)(w_{i+\mu}-w_k).
$$
Now we can rewrite the left hand side of \eqref{eq:helper_PoS_lower_3} as
$$
w^{-(i+\mu-k)(d+1)}\frac{\left[z_k+(\alpha+1)w_{i+\mu}-\alpha
w_k\right]^d-\left[z_k+(\alpha+1)w_{i+\mu}-(\alpha+1)w_k\right]^d}{(z_k+w_k)^d-z_k^d},
$$
and, if we additionally define for simplicity
$$
\zeta{:=} w^{\lambda},
\qquad\text{where}\qquad
\lambda{:=} \mu-k+i\in\ssets{1,\dots,\mu-2}
$$
and
$$
y_k{:=} \frac{z_k}{w_k},
$$
\eqref{eq:helper_PoS_lower_3} can be written as
$$
\zeta^{-(d+1)}\frac{\left[y_k+(\alpha+1)\zeta-\alpha
\right]^d-\left[y_k+(\alpha+1)\zeta-\alpha-1\right]^d}{(y_k+1)^d-y_k^d}\geq 1,
$$
or more simply,
$$
\zeta^{-(d+1)}f(x,y)\geq 1
$$
if we use function $f$ from \cref{lemma:ratio_monotone} with values
$$
x= (\alpha+1)w^{\lambda}-\alpha-1=(\alpha+1)(\zeta-1)>0
\qquad\text{and}\qquad y=y_k>0.
$$
Deploying the monotonicity of $f$ from \cref{lemma:ratio_monotone} and using that
$$
y=\frac{z_k}{w_k}=\alpha-\frac{1}{w_k}\sum_{j\in A} w_j\leq
\alpha-\frac{1}{w_k}w_{i+1}=\alpha-w^{i+1-k}=\alpha - w^{\lambda-\mu+1}\leq
\alpha-(1-\beta)\zeta,
$$
where the first inequality holds due to the fact that from the while-loop test in
line \ref{proc_line:dominate_while_1} of Procedure~\algoname{Dominate} we know that
$A\neq\emptyset$ and the last one because of $1-\beta=w^{-\mu}$, we finally get that it is enough if we show that
$$
\zeta^{d+1}\leq f\left((\alpha+1)(\zeta-1),\alpha-(1-\beta)\zeta\right)\qquad\text{for
all}\;\; \zeta\in\left[1,\frac{1}{1-\beta}\right],
$$
since $\zeta\geq w^1\geq 1$ and $\zeta\leq w^{\mu-2}\leq w^{\mu}=(1-\beta)^{-1}$.
But this is is satisfied, due to \eqref{eq:c_def_lemma_3} of \cref{lemma:ratio_monotone}, since we have selected our parameters $c=c_d$ and $\beta=\beta_d$ as in
\cref{lemma:def_c}.

\subsection{Proof of \texorpdfstring{\cref{eq:helper_PoS_lower_8}}{Equation (3.8)}}
\label{app:PoS_lower_8_proof}

In any such profile, player $i+\mu$ plays $
s_{i+\mu}^*$ and
\begin{itemize}
  \item Either all other players $j=i+1,\dots,i+\mu-1$ play $\tilde s_{j}$, in
  which case
  \begin{align*}
  C_i(\tilde s_i,\vecc s_{-i}') &= c_{i+\mu}\left(\sum_{\ell=i}^{i+\mu}
    w_\ell \right) + \sum_{j=i+1}^{i+\mu-1}c_j\left(\sum_{\ell=j-\mu}^{j-1}
    w_\ell\right)\\ &= c_{i+\mu}\left((\alpha+1)w^{i+\mu}\right) +
    \sum_{j=i+1}^{i+\mu-1}c_j\left(\alpha w^j\right)\\ &=
    w^{-(i+\mu)(d+1)}(\alpha+1)^dw^{d(i+\mu)} +
    \sum_{j=i+1}^{i+\mu-1}w^{-j(d+1)}\alpha^dw^{dj}\\
    &=w^{-(i+\mu)}(\alpha+1)^d+\sum_{j=i+1}^{i+\mu-1}w^{-j}\alpha^d\\
    &=w^{-i}\left[(\alpha+1)^dw^{-\mu}+\alpha^d\sum_{j=1}^{\mu-1}w^{-j}\right]\\
    &=w^{-i}\left[(\alpha+1)^dw^{-\mu}+\alpha^d(\alpha-w^{-\mu})\right],
  \end{align*}
  the last equality holding due to the definition of $\alpha$.
  \item Or there exists a \emph{single} player
  $k\in\ssets{i+1,\dots,i+\mu-1}$ that plays $s_k^*$ (instead of $\tilde s_k$
  which corresponds exactly to the previous case), in which case
  \begin{align*}
  C_i(\tilde s_i,\vecc s_{-i}') &\leq c_{k}\left(\sum_{\ell=k-\mu}^{k} w_\ell
    \right)+ c_{i+\mu}\left(\sum_{\ell=i}^{i+\mu} w_\ell -w_k \right) +
    \sum_{\underset{j\neq
    k,i+\mu}{j=i+1}}^{i+\mu}c_j\left(\sum_{\ell=j-\mu}^{j-1} w_\ell\right)\\ &=
    c_k\left((\alpha+1)w^{k} \right) +
    c_{i+\mu}\left((\alpha+1)w^{i+\mu}-w^k\right) + \sum_{\underset{j\neq
    k}{j=i+1}}^{i+\mu-1}c_j\left(\alpha w^j\right)\\
    &= w^{-i} \left[(\alpha+1)^dw^{-(k-i)}+(\alpha+1-w^{k-i-\mu})^dw^{-\mu}\right.\\
    &\qquad\qquad\qquad\qquad\qquad\qquad\left.+\alpha^d(\alpha-w^{-\mu}-w^{-(k-i)})\right],
  \end{align*} which is decreasing with respect to $k$, so taking the smallest
  possible value $k=i+1$ we have that
  \begin{multline*}
  C_i(\tilde s_i,\vecc s_{-i}') \leq \\
  w^{-i}\left[(\alpha+1)^dw^{-1}+(\alpha+1-w^{1-\mu})^dw^{-\mu}+\alpha^d(\alpha-w^{-\mu}-w^{-1})\right].
  \end{multline*}
\end{itemize} Considering both the above possible scenarios, in order to prove
\eqref{eq:helper_PoS_lower_8} it is thus sufficient to make sure that
\begin{equation}
\label{eq:helper_PoS_lower_10}
\alpha^d(\alpha-w^{-\mu})< (\alpha+1)^d(1-w^{-\mu})
\end{equation} and
\begin{equation}
\label{eq:helper_PoS_lower_11}
(\alpha+1-w^{1-\mu})^dw^{-\mu}+\alpha^d(\alpha-w^{-\mu}-w^{-1})<
(\alpha+1)^d(1-w^{-1}).
\end{equation} For \eqref{eq:helper_PoS_lower_10}, its left-hand side can be written
as
$$
(\beta\Phi_d)^d\left(\beta\Phi_d-(1-\beta)\right) <
(\beta\Phi_d)^d\left(\beta\Phi_d\right) =\beta^{d+1}\Phi_d^{d+1}
=\beta^{d+1}(\Phi_d+1)^{d},
$$
the first inequality holding because $\beta<1$, while the right-hand side is
\begin{equation*}
(\beta\Phi_d+1)^d(1-(1-\beta))=\beta^{d+1}\left(\Phi_d+\frac{1}{\beta}\right)^d.
\end{equation*} Thus it is enough to prove that
$$
\left(\Phi_d+1\right)^{d}\leq\left(\Phi_d+\frac{1}{\beta}\right)^d,
$$
which holds since $\beta\in (0,1)$.

For \eqref{eq:helper_PoS_lower_11}, the left-hand side is written as
\begin{align*}
&\left(\beta\Phi_d+1-\left(1+\frac{1}{\Phi_d}\right)(1-\beta)\right)^d(1-\beta)\\
&\qquad\qquad\qquad\qquad\qquad+ (\beta\Phi_d)^d\left(\beta\Phi_d-(1-\beta)-\left(1+\frac{1}{\Phi_d}\right)^{-1}\right)\\
  =& \left(\beta\Phi_d+\beta-\frac{1-\beta}{\Phi_d}\right)^d(1-\beta)
  +\beta^d\Phi_d^d\left(\beta\Phi_d-(1-\beta)-\frac{\Phi_d}{\Phi_d+1}\right)\\ <&
  \left(\beta\Phi_d+\beta\right)^d(1-\beta)
  +\beta^d\Phi_d^d\left(\beta\Phi_d-\frac{\Phi_d}{\Phi_d+1}\right)\\ =&
  \beta^d\left(\Phi_d+1\right)^d(1-\beta)
  +\beta^d\Phi_d^d\left(\beta\Phi_d-\frac{\Phi_d}{\Phi_d+1}\right)\\
  =&\beta^d(\Phi_d)^{d+1}(1-\beta)
  +\beta^d\Phi_d^d\left(\beta\Phi_d-\frac{\Phi_d}{\Phi_d+1}\right)\\
  =&\beta^d\Phi_d^d\left(\Phi_d-\beta\Phi_d+\beta\Phi_d-\frac{\Phi_d}{\Phi_d+1}\right)\\
  =&\beta^d\frac{\Phi_d^{d+2}}{\Phi_d+1}
\end{align*} and the right-hand side
$$
(\beta\Phi_d+1)^d\left(1-\left(1+\frac{1}{\Phi_d}\right)^{-1}\right)
=\beta^d\left(\Phi_d+\frac{1}{\beta} \right)^d\frac{1}{\Phi_d+1}.
$$
Thus it suffices to prove that
$$
\Phi_d^{d+2}\leq \left(\Phi_d+\frac{1}{\beta} \right)^d,
$$
which holds due to \eqref{eq:c_def_lemma_1} since we have already selected parameter
$c=c_d$ as in \eqref{eq:c_def}.

\section{Upper Bound Proofs}

\subsection{Technical Lemmas}

\begin{lemma}
\label{lemma:helper_1}
For any positive integer $m$ and real $x>0$,
$$
\left(1+\frac{1}{x}\right)^m \geq 1+\frac{m+1}{2x}
$$
\end{lemma}
\begin{proof}
Expanding the power in the left hand side we get
$$
\left(1+\frac{1}{x}\right)^m
= \sum_{j=0}^m\binom{m}{j}\frac{1}{x^j}
\geq \sum_{j=0}^1\binom{m}{j}\frac{1}{x^j}
=1+\frac{m}{x}\geq 1+ \frac{m+1}{2x},
$$
since
$$
\frac{m}{x}\geq \frac{m+1}{2x}
\ifif m \geq \frac{m+1}{2}
\ifif
m \geq 1.
$$
\end{proof}
\begin{lemma}
\label{lemma:helper_2}
For any integer $m\geq 0$ and real $x\geq 0$,
$$
(x+1)^{m+1}-x^{m+1} \leq \frac{m+1}{2}\left[(x+1)^m+x^m\right].
$$
\end{lemma}
\begin{proof}
Expanding the powers, our inequality can be rewritten equivalently as:
$$
\sum_{j=0}^{m+1}\binom{m+1}{j} x^j -x^{m+1}
\leq
\frac{m+1}{2}\left[\sum_{j=0}^m\binom{m}{j}x^j+x^m\right]
$$
$$
\sum_{j=0}^{m-1}\binom{m+1}{j} x^j+ (m+1)x^m
\leq
\frac{m+1}{2}\left[\sum_{j=0}^{m-1}\binom{m}{j}x^j+2x^m\right].
$$
$$
\sum_{j=0}^{m-1}\binom{m+1}{j} x^j
\leq
\frac{m+1}{2}\sum_{j=0}^{m-1}\binom{m}{j}x^j.
$$
Now, we can see that the above holds by bounding each term; for integers $j=0,1,\dots,m-1$:
\begin{align*}
 \binom{m+1}{j}
 &=
 \frac{(m+1)!}{(m+1-j)! j!}
 = \frac{m+1}{m+1-j}\frac{m!}{(m-j)! j!}\\
 &= \frac{m+1}{m+1-j}\binom{m}{j}
 \leq \frac{m+1}{2}\binom{m}{j}.
\end{align*}
\end{proof}

\subsection{Proof of \texorpdfstring{\cref{lemma:potential_mono_seq_monotonicity}}{Lemma~4.2}}
\label{app:potential_mono_seq_monotonicity_proof}

Observe that from the definition of $A_m$ in \eqref{eq:A_def},
$$
(A_m(x))^{-1}=\frac{1}{m+1}+\frac{1}{2x}
$$
which is decreasing with respect to $m$, and
$$
\left(\frac{A_m(x)}{m+1}\right)^{-1}
=1+\frac{m+1}{2x},
$$
which is increasing with respect to $m$. For the remaining sequence, observe that for any integer $m\geq 0$ and reals $y\geq x\geq 1$,
$$
\frac{A_m(x)}{A_m(y)}\geq \frac{A_{m+1}(x)}{A_{m+1}(y)}
\ifif
\frac{A_{m+1}(y)}{A_m(y)}\geq \frac{A_{m+1}(x)}{A_m(x)},
$$
so it is enough to show that function $\frac{A_{m+1}(x)}{A_m(x)}$ is monotonically increasing with respect to $x\geq 0$. Indeed,
\begin{align*}
\frac{A_{m+1}(x)}{A_m(x)}
&=\frac{\frac{1}{m+1}+\frac{1}{2x}}{\frac{1}{m+2}+\frac{1}{2x}}
=\frac{2x(m+2)+(m+1)(m+2)}{{2x(m+1)+(m+1)(m+2)}}\\
&=1+\frac{1}{m+1}\left(1+\frac{m+2}{2x}\right)^{-1}.
\end{align*}

\subsection{Proof of \texorpdfstring{\cref{lemma:potential_mono_bounds}}{Lemma~4.3}}
\label{app:potential_mono_bounds_proof}

First for \eqref{eq:potential_mono_bounds_2}, notice that it can be rewritten equivalently as
$$
\frac{1}{m+1}\leq\frac{S_m(\gamma x)}{(\gamma x)^{m+1}}=\frac{1}{A_m(\gamma x)} \leq \frac{1}{A_m(\gamma)},
$$
which holds, as an immediate consequence of the monotonicity of function $A_m$ (see \eqref{eq:A_def_boundaries}), given that $\gamma x\geq \gamma\geq 1$.
For \eqref{eq:potential_mono_bounds_1}, it is enough to prove just the special case when $w=1$, i.e.,
\begin{equation}
\label{eq:potential_mono_bounds_1_unweighted}
\frac{\gamma^{m+1}}{A_m(\gamma)}=S_m(\gamma)
\leq \frac{S_m(\gamma x+\gamma)-S_m(\gamma x)}{(x+1)^m}
\leq \gamma^{m+1},
\end{equation}
since then it is not difficult to check that we can recover the more general case in
\eqref{eq:potential_mono_bounds_1} by simply substituting $\gamma := \gamma w$ and
$x := \frac{x}{w}$ in \eqref{eq:potential_mono_bounds_1_unweighted}.

It is not difficult to check that \eqref{eq:potential_mono_bounds_1_unweighted}
holds for $m=0$, recalling that $S_0(x)=x$ for all $x\geq 0$. Next, assume for the
remainder of the proof that $m\geq 1$.

For the left-hand inequality of \eqref{eq:potential_mono_bounds_1_unweighted} first,
it can be equivalently rewritten as:
\begin{multline*}
(x+1)^m\left(\frac{1}{m+1}\gamma^{m+1}+\frac{1}{2}\gamma^m\right)
\leq \\
\frac{1}{m+1}\gamma^{m+1}\left[(x+1)^{m+1}-x^{m+1}\right]
\frac{1}{2}\gamma^m\left[(x+1)^m-x^m\right]
\end{multline*}
$$
\frac{1}{2}x^m \leq \frac{\gamma}{m+1}\left[(x+1)^{m+1}-(x+1)^m-x^{m+1} \right],
$$
and since $\gamma \geq 1$, it is sufficient to show that
$$
(m+1)x^m \leq 2\left[(x+1)^{m+1}-(x+1)^m-x^{m+1} \right]
$$
and thus, enough to show that
$$
(m+1)x^m \leq 2x\left[(x+1)^{m}-x^m \right].
$$
Now observe that the above trivially holds if $x=0$, while for $x>0$ it can be equivalently written as
$$
\frac{m+1}{2x} \leq \left(1+\frac{1}{x}\right)^m-1,
$$
which holds due to \cref{lemma:helper_1}.

For the right-hand inequality of \eqref{eq:potential_mono_bounds_1_unweighted}, it
can be equivalently written as:
$$
\frac{1}{m+1}\gamma\left[(x+1)^{m+1}-x^{m+1}\right]+\frac{1}{2}\left[(x+1)^m-x^m\right] \leq \gamma (x+1)^m
$$
$$
2\gamma \left[(x+1)^{m+1}-x^{m+1}\right] \leq (m+1)\left[(2\gamma-1)(x+1)^m+x^m\right]
$$
$$
(x+1)^{m+1}-x^{m+1} \leq (m+1)\left[\left(1-\frac{1}{2\gamma}\right)(x+1)^m+\frac{1}{2\gamma} x^m\right].
$$
Since $\gamma \geq 1$, we know that $\frac{1}{2\gamma}\in[0,\frac{1}{2}]$. Thus,
taking into consideration that $(x+1)^m>x^m\geq 0$, the linear combination on the
right-hand side of the above inequality is minimized for
$\frac{1}{2\gamma}=\frac{1}{2}$. So, it is enough to show that
$$
(x+1)^{m+1}-x^{m+1} \leq \frac{m+1}{2}\left[(x+1)^m+x^m\right],
$$
which holds due to \cref{lemma:helper_2}.

\subsection{Equivalence of \texorpdfstring{\cref{th:PoS_upper_general}}{Theorem~4.4} and \texorpdfstring{\cref{claim:PoS_upper_general}}{Claim~4.7}}
\label{app:restatement_main_upper}

To verify that \cref{claim:PoS_upper_general} gives indeed an equivalent restatement
of \cref{th:PoS_upper_general}, fix an arbitrary $W\geq 1$ and observe the
equivalence
\begin{equation*}
\alpha= A_d(\gamma W)=\frac{2(d+1)\gamma W}{2\gamma W+d+1}
\quad\ifif\quad
\gamma= \frac{1}{2W}\frac{\alpha(d+1)}{d+1-\alpha},
\end{equation*}
by using the definition of function $A_d$ from \eqref{eq:A_def}.
Therefore, it is not difficult to also compute that
\begin{align*}
\frac{d+1}{A_d(\gamma)}
&=(d+1)\left(\frac{1}{d+1}+\frac{1}{2\gamma}\right)=1+\frac{d+1}{2}\frac{1}{\gamma}\\
&= 1+\frac{d+1}{2}\cdot 2W \frac{d+1-\alpha}{\alpha(d+1)}
=1+W\left(\frac{d+1}{\alpha}-1\right).
\end{align*}

\section{Beyond Polynomial Latencies: Euler-Maclaurin}
\label{sec:EulerMaclaurin}
Our definition of the approximate potential function in
\cref{sec:faulhaber,sec:full_faulhaber_note} was based in Faulhaber's formula
\eqref{eq:faulhaber_formula_expand} for the sum of powers of positive integers. This
approach can be generalized further, by considering the Euler-Maclaurin summation
formula\footnote{See, e.g., \citep[Section 9.5]{Graham1989a} and \citep{Lehmer1940}.} :
\begin{equation}
\label{eq:eulermaclaurin}
\sum_{j=0}^{n} f(j)=\int_0^{n}f(t)\,dt+\frac{1}{2}[f(n)+f(0)]+\sum_{j=2}^m\frac{B_{j}}{j!}[f^{(j-1)}(n)-f^{(j-1)}(0)]+\mathcal{R}_m,
\end{equation}
for any infinitely differentiable function $f:[0,\infty)\map (0,\infty)$ (with
$f^{(j)}$ denoting the $j$-th order derivative of $f$) and integers $n,m\geq 1$,
where $B_j$ denotes the Bernoulli numbers we have already used in
\cref{sec:faulhaber} and the \emph{error-term} $\mathcal R_{m}$ can be bounded by
\begin{equation}
\label{eq:eulermaclaurin_error_bound}
\card{\mathcal{R}_m}\leq \frac{2\zeta(m)}{(2\pi)^m}\int_0^{n}\card{f^{(m)}(t)}\,dt,
\end{equation}
where $\zeta(m)=\sum_{j=1}^\infty\frac{1}{j^m}$ is Riemann's zeta function. Thus, if
function $f$ is such that the quantity in the right-hand side of
\eqref{eq:eulermaclaurin_error_bound} eventually vanishes, i.e. for any real $x\geq
0$,
\begin{equation}
\label{eq:eulermaclaurin_error_vanish}
\lim_{m\to\infty} \frac{\zeta(m)}{(2\pi)^m}\int_0^{x}\card{f^{(m)}(t)}\,dt =0,
\end{equation}
then we can define our approximate-potential candidate function on any real $x\geq 0$ by generalizing \eqref{eq:eulermaclaurin}:
\begin{equation}
\label{eq:pot_eulermaclaurin}
S(x)=S_f(x) = \int_0^{x}f(t)\,dt+\frac{1}{2}[f(x)+f(0)]+\sum_{j=2}^\infty\frac{B_{j}}{j!}[f^{(j-1)}(x)-f^{(j-1)}(0)].
\end{equation}
For example, it is not difficult to see that, for any monomial $f(x)=x^d$ of degree
$d\geq 1$, condition \eqref{eq:eulermaclaurin_error_vanish} is indeed satisfied
(since $f^{(m)}=0$ for all $m\geq d+1$) and, because also $f^{(m)}(0)=0$ and
$f^{(m)}(x)=\frac{d!}{(d-m)!}x^{d-m}$, one recovers exactly
\eqref{eq:faulhaber_formula_expand} from \eqref{eq:pot_eulermaclaurin} above.

Let us now demonstrate this general approach for latency functions $f$ that are not
polynomials. For the remaining of this section let $f(x)=e^x$ be an exponential
delay function. Then, for any $y\geq 0$,
$$
\lim_{m\to\infty}\frac{\zeta(m)}{(2\pi)^m}\int_0^{y}\card{f^{(m)}(t)}\,dt =(e^y-1) \lim_{m\to\infty}\frac{\zeta(m)}{(2\pi)^m}=0,
$$
since $\lim_{m\to\infty}\zeta(m)=1$ and $\lim_{m\to\infty}(2\pi)^m=\infty$. Thus, condition \eqref{eq:eulermaclaurin_error_vanish} is satisfied, and we can define from \eqref{eq:pot_eulermaclaurin}
\begin{align*}
S(x)&= \int_0^{x}e^t\,dt+\frac{1}{2}[e^x+e^0]+\sum_{j=2}^\infty\frac{B_{j}}{j!}[e^x-e^0]\\
&= (e^x-1)-\frac{1}{2}(e^x-1)+\sum_{j=2}^\infty\frac{B_{j}}{j!}(e^x-1) +e^x\\
&= (e^x-1)\sum_{j=0}^\infty\frac{B_j}{j!}+e^x.
\end{align*}
But since for the integer value $x=1$ we know that
$
S(1)=\sum_{j=0}^1f(j)=1+e,
$
it must be that
$$
e+1=(e^x-1)\sum_{j=0}^\infty\frac{B_j}{j!}+e
\quad\ifif\quad
\sum_{j=0}^\infty\frac{B_j}{j!}=\frac{1}{e-1}.
$$
So, we
finally have that
$$
S(x)=(e^x-1)\frac{1}{e-1}+e^x=\frac{e^{x+1}-1}{e-1}.
$$
From this, for all reals $x\geq 0$, $w>0$ we compute:
\begin{equation}
\label{eq:expo_potential}
\frac{S(x+w)-S(x)}{w f(x+w)}=\frac{1}{e-1}\frac{e^{x+w+1}-e^{x+1}}{w e^{x+w}}=\frac{e}{e-1}\frac{1-e^{-w}}{w},
\end{equation}
which does not depend on $x$. Thus, from \eqref{eq:cond_approx_eq} in
\cref{lemma:approximate_PoS_ratios} we deduce that \emph{exact} pure Nash equilibria
always exist for weighted congestion games with exponential latencies. The function
$S$ we defined in \eqref{eq:expo_potential} essentially serves as a weighted
potential \citep{Monderer:1996sp}; its global minimum is a pure Nash equilibrium.
Notice here that these results regarding exponential latency functions were already
known by the work of \citet{Panagopoulou2007}.

\section{Social Optimum is a \texorpdfstring{($d+1$)}{(d+1)}-Approximate Equilibrium}
In this section we show that the socially optimum solution is itself an
$(d+1)$-appproximate equilibrium, where $d$ is the maximum degree of the polynomial
latency functions. We must mention here that, the approximation factor on its own,
i.e.\ $d+1$, does not constitute a novel contribution: existence of
$(d+1)$-approximate equilibria was already known by the work of \citet{Harks2012a}.
However, the new element in \cref{th:PoS_OPT} below is that this can be achieved by
all optimum solutions.

Related to this, we would like to emphasize that, if we only cared about showing the
existence of \emph{an} optimal solution that is a $(d+1)$-approximate equilibrium,
this would have been an immediate corollary of our main upper bound result: by
simply setting $\alpha=d+1$ in \cref{th:PoS_upper_general} we get exactly what we
want. However, the following theorem demonstrates the stronger statement that
\emph{all} social cost minimizers have the property we want.

\begin{theorem}
\label{th:PoS_OPT}
Consider any weighted congestion game with polynomial latency functions of maximum
degree $d$ and let $\vecc{s}^*$ be a strategy profile that minimizes social cost.
Then $\vecc{s}^*$ is a $(d+1)$--approximate pure Nash equilibrium. As an immediate
consequence, the Price of Stability of $(d+1)$--approximate Nash equilibria is $1$.
\end{theorem}
\begin{proof}
Let $c$ be an arbitrary cost function of maximum degree $d$ with non-negative
coefficients, i.e., $c(x)=\sum_{j=0}^d a_j x^j$, with $a_j\ge 0$ for all $j$. We
will first show that for all $w>0$ and $x\ge 0$:
\begin{align}\label{eq:optimum}
w \cdot c(x+w) \leq (x+w) \cdot c(x+w) - x \cdot c(x) \leq (d+1) \cdot w \cdot c(x+w).
\end{align}
To this end, with $z=\frac{x}{w}$, we get
\begin{align*}
(x+w) \cdot c(x+w) - x \cdot c(x)
&= \sum_{j=0}^d a_j \cdot \left[ (x+w)^{j+1} - x^{j+1} \right]\\
&= \sum_{j=0}^d a_j \cdot w^{j+1}\left[ (1+z)^{j+1} - z^{j+1} \right]\\
&= \sum_{j=0}^d a_j \cdot w^{j+1}\left[ \sum_{k=0}^j \binom{j+1}{k} z^k \right],
\end{align*}
and
\begin{align*}
w \cdot c(x+w)
&= \sum_{j=0}^d a_j \cdot w(x+w)^{j}\\
&= \sum_{j=0}^d a_j \cdot w^{j+1} \left[\sum_{k=0}^j \binom{j}{k} z^k \right].
\end{align*}
Clearly, $\binom{j+1}{k}\geq\binom{j}{k}$ for all integer $j\in[0,d], k\in[0,j]$,
which immediately implies the first inequality in \eqref{eq:optimum}. To see the
second inequality, observe that
\begin{align*}
\frac{\binom{j+1}{k}}{\binom{j}{k}} = \frac{j+1}{j+1-k} \le j+1 \le d+1.
\end{align*}

Since $\vecc{s}^*$ minimizes social cost, for all players $i\in[n]$ and strategies
$s_i\in S_i$,
\begin{align*}
C(\vecc{s}^*) \leq C(s_i,\vecc{s}^*_{-i}).
\end{align*}

Denoting $y_e=\sum_{j\in[n]\setminus\{i\}: e\in s^*_j} w_j$,
from \eqref{eq:optimum}, we get
\begin{align*}
 0 &\leq C(s_i,\vecc{s}^*_{-i}) -C(\vecc{s}^*) \\
 &= \sum_{e\in s_i\setminus s_i^*} \left[ (y_e+w_i) c_e(y_e+w_i) - y_e c_e(y_e) \right]\\
  &\qquad\qquad\qquad\qquad\qquad\qquad\qquad- \sum_{e\in s^*_i\setminus s_i}  \left[  (y_e+w_i) c_e(y_e+w_i)- y_e c_e(y_e) \right]\\
&= \sum_{e\in s_i} \left[ (y_e+w_i) c_e(y_e+w_i) - y_e c_e(y_e) \right]
  - \sum_{e\in s^*_i}  \left[  (y_e+w_i) c_e(y_e+w_i)- y_e c_e(y_e) \right]\\
 & \leq  (d+1) \sum_{e\in s_i} w_i  \cdot c_e(y_e+w_i) -  \sum_{e\in s^*_i} w_i  \cdot c_e(y_e+w_i)\\
 &= (d+1) C_i(s_i,\vecc{s}^*_{-i}) -C_i(\vecc{s}^*),
\end{align*}
or equivalently  $C_i(\vecc{s}^*)\leq (d+1) C_i(s_i,\vecc{s}^*_{-i})$. So $\vecc{s}^*$ is a $(d+1)$-approximate Nash equilibrium.
\end{proof}

\section*{Acknowledgments} We thank the anonymous reviewers for the careful and thorough reading of our manuscript, and for their valuable feedback.

\bibliographystyle{abbrvnat} 
\bibliography{weightedPoS}
\end{document}